\documentclass[journal,two column]{IEEEtran}
\IEEEoverridecommandlockouts
\usepackage{amsmath,amsfonts}
\usepackage{algorithmic}
\usepackage{algorithm}
\usepackage{array}
\usepackage{caption}
\usepackage[caption=false,font=footnotesize,labelfont=rm,textfont=rm]{subfig}
\DeclareSubrefFormat{myformat}{#1(#2)}
\usepackage{textcomp}
\usepackage{stfloats}
\usepackage{url}
\usepackage{color}
\usepackage{verbatim}
\usepackage{graphicx}
\usepackage{cite}
\usepackage{tabularx}  % 自动调整列宽
\usepackage{siunitx} % 用于单位格式化
\usepackage{booktabs}   % 专业表格线
\usepackage{amsmath}   % 数学符号支持
\usepackage{hyperref}
\hypersetup{colorlinks=true, citecolor=blue}
\usepackage{amssymb}
\usepackage{psfrag}
\usepackage{epstopdf}
\usepackage{authblk}
\usepackage{amsthm} 
\usepackage{lipsum}
\usepackage{mathtools}
\usepackage{cuted}
\usepackage{threeparttable}
\usepackage{hyperref}
\hypersetup{colorlinks=true, citecolor=blue}
\usepackage{adjustbox}
\usepackage{multirow}
\makeatother

\usepackage{xpatch}  
\makeatletter  
\xpatchcmd{\@thm}{\thm@headpunct{.}}{\thm@headpunct{}}{}{}

\begin{document}
% Bold lowercase: syntax \nb# where # is {a ... z, 0,1}, bf means bold fonts
\def\mba{{\mathbf{a}}} %% 小写加粗
\def\mbb{{\mathbf{b}}}
\def\mbc{{\mathbf{c}}}
\def\mbd{{\mathbf{d}}}
\def\mbe{{\mathbf{e}}}
\def\mbf{{\mathbf{f}}}
\def\mbg{{\mathbf{g}}}
\def\mbh{{\mathbf{h}}}
\def\mbi{{\mathbf{i}}}
\def\mbj{{\mathbf{j}}}
\def\mbk{{\mathbf{k}}}
\def\mbl{{\mathbf{l}}}
\def\mbm{{\mathbf{m}}}
\def\mbn{{\mathbf{n}}}
\def\mbo{{\mathbf{o}}}
\def\mbp{{\mathbf{p}}}
\def\mbq{{\mathbf{q}}}
\def\mbr{{\mathbf{r}}}
\def\mbs{{\mathbf{s}}}
\def\mbt{{\mathbf{t}}}
\def\mbu{{\mathbf{u}}}
\def\mbv{{\mathbf{v}}}
\def\mbw{{\mathbf{w}}}
\def\mbx{{\mathbf{x}}}
\def\mby{{\mathbf{y}}}
\def\mbz{{\mathbf{z}}}
\def\mb0{{\mathbf{0}}}
\def\mb1{{\mathbf{1}}}

% Bold capital letters: syntax \nb# where # is {A ... Z}
\def\mbA{{\mathbf{A}}} %% 大写加粗
\def\mbB{{\mathbf{B}}}
\def\mbC{{\mathbf{C}}}
\def\mbD{{\mathbf{D}}}
\def\mbE{{\mathbf{E}}}
\def\mbF{{\mathbf{F}}}
\def\mbG{{\mathbf{G}}}
\def\mbH{{\mathbf{H}}}
\def\mbI{{\mathbf{I}}}
\def\mbJ{{\mathbf{J}}}
\def\mbK{{\mathbf{K}}}
\def\mbL{{\mathbf{L}}}
\def\mbM{{\mathbf{M}}}
\def\mbN{{\mathbf{N}}}
\def\mbO{{\mathbf{O}}}
\def\mbP{{\mathbf{P}}}
\def\mbQ{{\mathbf{Q}}}
\def\mbR{{\mathbf{R}}}
\def\mbS{{\mathbf{S}}}
\def\mbT{{\mathbf{T}}}
\def\mbU{{\mathbf{U}}}
\def\mbV{{\mathbf{V}}}
\def\mbW{{\mathbf{W}}}
\def\mbX{{\mathbf{X}}}
\def\mbY{{\mathbf{Y}}}
\def\mbZ{{\mathbf{Z}}}

% \mathcal: syntax \mc# where # is {A ... Z}, cal means calligraphy
\def\mcA{{\mathcal{A}}}
\def\mcB{{\mathcal{B}}}
\def\mcC{{\mathcal{C}}}
\def\mcD{{\mathcal{D}}}
\def\mcE{{\mathcal{E}}}
\def\mcF{{\mathcal{F}}}
\def\mcG{{\mathcal{G}}}
\def\mcH{{\mathcal{H}}}
\def\mcI{{\mathcal{I}}}
\def\mcJ{{\mathcal{J}}}
\def\mcK{{\mathcal{K}}}
\def\mcL{{\mathcal{L}}}
\def\mcM{{\mathcal{M}}}
\def\mcN{{\mathcal{N}}}
\def\mcO{{\mathcal{O}}}
\def\mcP{{\mathcal{P}}}
\def\mcQ{{\mathcal{Q}}}
\def\mcR{{\mathcal{R}}}
\def\mcS{{\mathcal{S}}}
\def\mcT{{\mathcal{T}}}
\def\mcU{{\mathcal{U}}}
\def\mcV{{\mathcal{V}}}
\def\mcW{{\mathcal{W}}}
\def\mcX{{\mathcal{X}}}
\def\mcY{{\mathcal{Y}}}
\def\ncalZ{{\mathcal{Z}}}

% \mathbb: syntax \mbb# where # is {A ... Z}, bb means blackboard bold
\def\mbbA{{\mathbb{A}}}
\def\mbbB{{\mathbb{B}}}
\def\mbbC{{\mathbb{C}}}
\def\mbbD{{\mathbb{D}}}
\def\mbbE{{\mathbb{E}}}
\def\mbbF{{\mathbb{F}}}
\def\mbbG{{\mathbb{G}}}
\def\mbbH{{\mathbb{H}}}
\def\mbbI{{\mathbb{I}}}
\def\mbbJ{{\mathbb{J}}}
\def\mbbK{{\mathbb{K}}}
\def\mbbL{{\mathbb{L}}}
\def\mbbM{{\mathbb{M}}}
\def\mbbN{{\mathbb{N}}}
\def\mbbO{{\mathbb{O}}}
\def\mbbP{{\mathbb{P}}}
\def\mbbQ{{\mathbb{Q}}}
\def\mbbR{{\mathbb{R}}}
\def\mbbS{{\mathbb{S}}}
\def\mbbT{{\mathbb{T}}}
\def\mbbU{{\mathbb{U}}}
\def\mbbV{{\mathbb{V}}}
\def\mbbW{{\mathbb{W}}}
\def\mbbX{{\mathbb{X}}}
\def\mbbY{{\mathbb{Y}}}
\def\nbbZ{{\mathbb{Z}}}

% \mathfrak: math fraktur
\def\mfrakR{{\mathfrak{R}}}

% Roman: {\rm } syntax \mrm# where # is {a ... z}, rm means roman
\def\mrma{{\rm a}}
\def\mrmb{{\rm b}}
\def\mrmc{{\rm c}}
\def\mrmd{{\rm d}}
\def\mrme{{\rm e}}
\def\mrmf{{\rm f}}
\def\mrmg{{\rm g}}
\def\mrmh{{\rm h}}
\def\mrmi{{\rm i}}
\def\mrmj{{\rm j}}
\def\mrmk{{\rm k}}
\def\mrml{{\rm l}}
\def\mrmm{{\rm m}}
\def\mrmn{{\rm n}}
\def\mrmo{{\rm o}}
\def\mrmp{{\rm p}}
\def\mrmq{{\rm q}}
\def\mrmr{{\rm r}}
\def\mrms{{\rm s}}
\def\mrmt{{\rm t}}
\def\mrmu{{\rm u}}
\def\mrmv{{\rm v}}
\def\mrmw{{\rm w}}
\def\mrmx{{\rm x}}
\def\mrmy{{\rm y}}
\def\mrmz{{\rm z}}

% Theorems etc.
\newtheorem{lemma}{Lemma}
\newtheorem{definition}{Definition}
\newtheorem{remark}{Remark}
\newtheorem{theorem}{Theorem}
\newtheorem{proposition}{Proposition}
\newtheorem{corollary}{Corollary}
\newtheorem{example}{Example}
\newtheorem{assumption}{Assumption}

%%%
\newcommand{\upcite}[1]{\textsuperscript{\textsuperscript{\cite{#1}}}} 
\newcommand{\romann}[1]{\uppercase\expandafter{\romannumeral #1}} 
\newcommand{\Romann}[1]{\expandafter{\romannumeral #1}} 
\newcommand{\colorb}[1]{{\color{blue} #1}} 
\newcommand{\colorr}[1]{{\color{red} #1}}

%%%%%%%% Backwards compatibility

\newcommand{\ceil}[1]{\lceil #1\rceil} % ceiling

\def\argmin{\operatorname{arg~min}}
\def\argmax{\operatorname{arg~max}}

% sinc, cosc
\def\sinc{{\rm sinc}}
\def\cosc{{\rm cosc}}

% special symbol
\def\larrow{\leftarrow}
\def\rarrow{\rightarrow}
\def\triequ{\triangleq}
\def\simequ{\simeq}

\def\figref#1{Fig.\,\ref{#1}}%
\def\tabref#1{Table\,\ref{#1}}%
\def\equref#1{(\ref{#1})}%
\def\appref#1{Appendix\,\ref{#1}}%
\def\lemref#1{Lemma\,\ref{#1}}%
\def\defref#1{Definition\,\ref{#1}}%
\def\theref#1{Theorem\,\ref{#1}}%
\def\remref#1{Remark\,\ref{#1}}%
\def\secref#1{Sec.\,\ref{#1}}%

\def\assref#1{Assumption\,\ref{#1}}%
\def\ie{{\em i.e.}}
\def\eg{{\em e.g.}}
\def\rme{{\rm e}}
\def\rmd{{\rm d}}
\def\dB{{\rm dB}}
\def\x{\times}
\def\T{\intercal}
\def\H{\dagger}
\def\wbar{\overline}
\def\what{\widehat}
\def\d{{\rm d}}
\def\E{{\mathbb E}}
\def\pd{\partial}
\def\e{{\rm e}}
\def\1{\mathbbmtt{1}}
\def\var{\operatorname{Var}}
\def\cov{\operatorname{Cov}}
\def\mean{\operatorname{mean}}
\def\P{{\mathbb P}}

\def\R{{\mathbb R}}

\def\erfc{\operatorname{erfc}}
\def\erf{\operatorname{erf}}
\def\opt{\mathrm{opt}}

\def\sinr{\mathtt{SINR}}   % Signal to interference plus noise ratio
\def\snr{\mathtt{SNR}}
\def\sir{\mathtt{SIR}}
\def\scnr{\mathtt{SCNR}}
\def\ase{\mathtt{ASE}}
\def\se{\mathtt{SE}}

\title{Offset Pointing for Energy-efficient Reception in Underwater Optical Wireless Communication: Modeling and Performance Analysis}
\author{
Qiyu~Ma,~\IEEEmembership{Student Member,~IEEE, }
Jiajie Xu,~\IEEEmembership{ Member, IEEE} 
and~Mohamed-Slim~Alouini,~\IEEEmembership{Fellow,~IEEE}
\thanks{Qiyu Ma is with the Department of Electronic Engineering, Tsinghua University, Beijing, 100084, China (Email: ma-qy22@mails.tsinghua.edu.cn), Jiajie Xu and Mohamed-Slim Alouini are  with the Communication Theory Lab (CTL), Computer, Electrical, and Mathematical Science and Engineering (CEMSE) division in King Abdullah University of Science and Technology (KAUST), Thuwal 23955, Makkah Province, Saudi Arabia. (Emails: jiajie.xu.1@kaust.edu.sa; slim.alouini@kaust.edu.sa)\\

% Copyright (c) 2026 IEEE. Personal use of this material is permitted. However, permission to use this material for any other purposes must be obtained from the IEEE by sending a request to pubs-permissions@ieee.org.

}
}

\markboth{ }%
{}
\IEEEpubid{}

\maketitle

\begin{abstract}
Underwater Wireless Optical Communication (UOWC) is a key enabling technology for future space-air-ground-sea integrated networks. However, UOWC faces critical hurdles from spatial randomness and stringent energy constraints. These challenges fundamentally limit network lifetime and sustainability. This paper develops a comprehensive stochastic geometry framework to perform a differential energy analysis of UOWC links.Instead of relying on simplified models, we employ a three-dimensional truncated Poisson point process (TPPP) to accurately capture the anisotropic nature of the underwater environment, specifically the disparity between horizontal spread and vertical depth. It incorporates a Lambertian emission pattern, random receiver positions and orientations, and a realistic channel model with extinction effects. Under this model, we derive a full suite of closed-form expressions for key performance indicators. These include the nearest-neighbor distance distribution, expected received power, signal-to-noise ratio (SNR), and bit error rate (BER). A principal and counter-intuitive finding of our analysis is an ``offset-pointing'' strategy. This strategy involves intentionally misaligning the receiver by a deterministically optimal angle. This approach maximizes the integrated received power across the aperture, contrary to the conventional pursuit of perfect alignment. We formulate and solve an energy-efficiency optimization problem. Our results demonstrate that this strategy enhances system robustness and yields substantial performance gains. Simulation results validate our analytical models. They show that the optimal offset strategy can reduce the required transmit power by nearly 20\% to achieve a target BER. This reduction directly translates into extended network lifetime and higher total data throughput. These findings offer a new design paradigm for deploying robust, cost-effective, and sustainable UOWC networks.
\end{abstract}


\section{Introduction}
    \subsection{Background}
        The convergence of sixth-generation (6G) networks, which promises a future of the Internet of Everything \cite{dang2020should}, and the strategic vision for integrated space-air-ground-sea connectivity \cite{xu2023space}, has intensified the demand for efficient and reliable underwater wireless communication (UWC) technologies. Among the primary UWC modalities, radio frequency (RF) \cite{1353478}, acoustic \cite{4752682,6605638,han2018joint}, and optical \cite{saeed2019underwater,kaushal2016underwater,zeng2016survey}, each presents a distinct trade-off. RF signals suffer from prohibitive attenuation in seawater, severely limiting their range. While underwater acoustic communication enables long-distance transmission and is widely adopted, its inherent low bandwidth and high latency impede its ability to support high-throughput applications like video streaming. Consequently, underwater optical wireless communication (UOWC) emerges as a compelling alternative. By exploiting the low-attenuation blue-green optical window, UOWC can achieve transmission distances of up to 100 meters \cite{wang2019100,chen2020toward} while delivering data rates at the Gbps level \cite{hong2019discrete,lu201960m,liu201734}. These capabilities position UOWC as a critical enabling technology for future applications, including the underwater Internet of Things, collaborative drone swarms, and high-speed access for seabed observation networks.
        
        However, the distinctive characteristics of the underwater environment pose significant challenges to optical signal transmission. Severe path loss fundamentally restricts optical propagation, imposing a critical trade-off between transmission range and coverage in UOWC link design \cite{christopoulou2020optimal}. System performance is further degraded by a host of factors, including pointing errors between the transmitter (Tx) and receiver (Rx) \cite{9791364, ijeh2021parameter}, air bubbles from breaking waves \cite{oubei2017performance}, and oceanic turbulence \cite{zedini2019unified, lou2022performance}. Critically, the random displacement of terminal nodes, driven by ocean currents and wind friction, renders UOWC systems inherently stochastic and challenges conventional Pointing, Acquisition, and Tracking (PAT) mechanisms. While stochastic geometry (SG) offers a powerful analytical framework for such random distributions \cite{chiu2013stochastic, elsawy2016modeling}, its results often provide an averaged, macroscopic view. This paper pioneers a differential analysis framework to capture more nuanced performance implications. By differentiating the optical link along two fundamental degrees of freedom, the receiver's position and orientation, we can precisely quantify how infinitesimal changes in these variables impact key performance indicators (KPIs) such as the nearest neighbor distance distribution, received power, signal to noise ratio (SNR), and bit error rate (BER).
    \vspace{-4mm}
\subsection{Related Work}
    Modeling the stochastic uncertainties inherent in the dynamic underwater environment is critical for UOWC system analysis. Research has progressed along two main fronts: macroscopic network modeling and microscopic link-level analysis. On the macroscopic scale, stochastic geometry has been employed to analyze the impact of random node locations. Foundational studies provided insights into performance under receiver position uncertainty, first in 2D planes \cite{christopoulou2020performance} and later in 3D conical volumes \cite{vaiopoulos2023uowc}. Concurrently, microscopic analyses have focused on the physical layer, providing detailed models for the impact of pointing errors (i.e., orientation uncertainty), often in conjunction with other channel impairments like oceanic turbulence \cite{boluda2020impact, rahman2023performance, han2022study}.
    
    Despite these advances, two significant gaps persist in the literature. Firstly, while existing models analyze the effects of random position or orientation, they typically offer a probabilistic or averaged assessment of performance. A granular, differential framework capable of capturing the continuous, fine-grained impact of these two fundamental degrees of freedom on key performance indicators remains undeveloped. Secondly, and perhaps more critically, there is a notable decoupling of these physical uncertainty models from the analysis of node energy dynamics. While the sophisticated modeling of energy efficiency and differentiated power consumption is a mature research area in terrestrial wireless networks \cite{alamu2020energy, portillo2024energy}, a comprehensive statistical analysis that integrates energy constraints is largely absent in UOWC performance literature. This omission overlooks the paramount importance of network lifetime and sustainability for energy-constrained underwater applications.

%%%%%%%%%%%%%%%%%%%%%%%%%%%%%%%%%%%%%%%%%%%%%%%%%%%%%%%%%%%%%%%%%
\begin{table*}
\centering
\caption{Key Symbols and Definitions}\label{tab:revised_parameters}
\begin{tabularx}{\textwidth}{@{} c X c X @{}}
\toprule
\textbf{Symbol} & \textbf{Definition} & \textbf{Symbol} & \textbf{Definition} \\
\midrule
\multicolumn{4}{l}{\textit{\textbf{Stochastic Geometry \& System Model}}} \\
$\Phi_{3D}$ & 3D non-uniform point process representing node locations & $\Lambda$ & Intensity of the 2D PPP [\unit{nodes/m^2}] \\
$\mathcal{N}$ & Intensity of the 3D PPP in the slab, $\mathcal{N} = \Lambda/R$ [\unit{nodes/m^3}] & $R$ & Maximum vertical depth of the underwater environment [\unit{m}] \\
$L_{dis}$ & Euclidean distance between the Tx and Rx [\unit{m}] & $L_{deep}$ & Vertical depth of the receiver [\unit{m}] \\
$\theta$ & Angle of irradiance from the Tx optical axis & $\psi$ & Angle of incidence at the Rx optical axis \\
$\delta$ & Pointing error offset angle [\unit{\degree}] & $\delta_{opt}$ & Optimal pointing error offset angle [\unit{\degree}] \\
\midrule
\multicolumn{4}{l}{\textit{\textbf{Transmitter \& Channel Parameters}}} \\
$P_{Tx}$ & Transmitted optical power [\unit{W}] & $P_{Rx}$ & Received optical power [\unit{W}] \\
$\phi_{1/2}$ & Transmitter's half-power semi-angle [\unit{\degree}] & $m$ & Lambertian emission order, $m = -\ln(2)/\ln(\cos(\phi_{1/2}))$ \\
$\mathcal{H}$ & Line-of-sight (LOS) channel DC gain & $\mathfrak{S}$ & Radiant intensity distribution of the transmitter \\
$\mathcal{L}$ & Path loss due to medium extinction & $\mathfrak{R}$ & Receiver's response strength \\
$c(\lambda, z)$ & Wavelength and depth-dependent extinction coefficient [\unit{m^{-1}}] & $n$ & Refractive index of the receiver's optical concentrator \\
\midrule
\multicolumn{4}{l}{\textit{\textbf{Receiver \& Noise Parameters}}} \\
$A_r$ & Area of the receiver's aperture lens, $A_r = \pi D^2/4$ [\unit{m^2}] & $D$ & Diameter of the receiver's aperture lens [\unit{cm}] \\
$\phi_{FoV}$ & Receiver's field-of-view semi-angle [\unit{\degree}] & $T_s(\psi)$ & Transmittance of the optical filter \\
$\eta$ & Photon Detection Efficiency (PDE) of the SiPM & $G$ & Gain of the SiPM \\
$I_{SiPM}$ & Photocurrent generated by the SiPM [\unit{A}] & $I_d$ & Dark current of the SiPM [\unit{nA}] \\
$P_{ct}$ & Crosstalk probability of the SiPM & $F$ & Excess noise factor of the SiPM \\
$R_L$ & Load resistance [\unit{\Omega}] & $\Delta\lambda$ & Bandwidth of the optical bandpass filter [\unit{nm}] \\
$P_{sun}$ & Received solar background radiation power [\unit{W}] & $E_{sun}(z=0)$ & Solar irradiance at sea level [\unit{W/m^2}] \\
$\epsilon$ & Attenuation coefficient for solar radiation [\unit{m^{-1}}] & $\zeta_r$ & Reflectance factor for solar radiation \\
$\sigma_q^2$ & Quantum (shot) noise variance & $\sigma_d^2$ & Dark current noise variance \\
$\sigma_{th}^2$ & Thermal noise variance & $\sigma_{solar}^2$ & Solar background noise variance \\
\midrule
\multicolumn{4}{l}{\textit{\textbf{Performance \& Optimization Metrics}}} \\
$\text{SNR}$ & Signal-to-Noise Ratio & $\text{BER}$ & Bit Error Rate \\
$R_b$ & Achievable bit rate [\unit{bps}] & $B$ & System bandwidth [\unit{Hz}] \\
$E_{total}$ & Total energy budget for the network [\unit{J}] & $N_b$ & Total number of transmitted bits \\
$T_{tx}$ & Transmission time per node [\unit{s}] & $\mathcal{A}$ & Area of the network deployment region [\unit{m^2}] \\
\bottomrule
\vspace{-6mm}
\end{tabularx}
\end{table*}

%%%%%%%%%%%%%%%%%%%%%%%%%%%%%%%%%%%%%%%%%%%%%%%%%%%%%%%%%%%%%%%%%
\vspace{-3mm}
\subsection{Contribution}
            To bridge these critical gaps, this study develops an SG framework for the differential energy analysis of individual nodes in UOWC, explicitly accounting for randomized node positions and orientations. We investigate the performance impact of such randomness and aim to determine optimal node deployment and orientation strategies, thereby enabling a fundamental understanding of network lifetime and sustainability under spatial uncertainty. Our key contributions are threefold:
            \subsubsection{Stochastic Underwater Channel Modeling}
            We develop a specialized stochastic geometry framework for UOWC networks. By modeling the node distribution within an infinite slab using a three-dimensional truncated Poisson point process (TPPP), we capture the inherent geometric anisotropy of the underwater environment. While TPPP is an established mathematical tool, its specific application here provides a rigorous basis for the subsequent differential energy analysis, explicitly incorporating the receiver's random spatial position and three-dimensional orientation. Based on this, we establish a differential energy analysis framework to evaluate a single communication link, considering the transmitter's Lambertian emission and the receiver's stochastic properties.
            
            \subsubsection{Analytical Performance Expressions}
            We derive a comprehensive set of closed-form expressions for KPIs. These include the probability density function (PDF) and cumulative distribution function (CDF) for nearest-neighbor distance, the expected received optical power under random orientation, the SNR, and the BER. These analytical solutions enable performance evaluation without relying on extensive simulations.
            
            \subsubsection{Network Design and Optimization Insights}
            We formulate and solve an optimization problem to maximize energy efficiency (Bits per Joule), revealing a key trade-off between node density and transmission power. A principal finding is the significant performance gain achieved by an \textit{offset-pointing} strategy, where intentionally misaligning the receiver by an optimal angle increases received power and reduces BER. This offers practical guidance for deploying robust and energy-efficient UOWC networks.
            
\vspace{-4mm}           
\subsection{Organization}
    The remainder of this paper is structured as follows: Section \ref{sec:SM} details our analysis system model. In section \ref{sec:PE}, the key distribution and performance metrics are derived. Section \ref{sec:SR} analyzes experimental configurations and results. Conclusions are discussed in Section \ref{Sec:Conc}. The notation used throughout this paper is summarized in Table \ref{tab:revised_parameters}.

\section{System Model}\label{sec:SM}

This section details the comprehensive model for an UOWC system. We first describe the geometric link configuration and the stochastic model for node distribution. Subsequently, we formulate the channel model, encompassing the optical extinction coefficient and the end-to-end channel gain.

%%%%%%%%%%%%%%%%%%%%%%%%%%%%%%%%%%%%%%%%%%%%%%%%%%%%%%%%%%%%%%%%%
    \subsection{Stochastic Geometry Model}
    \label{2.a}
        To accurately model the three-dimensional and anisotropic characteristics of the node distribution, we consider a network deployed within an infinite slab defined as $S = \mathbb{R}^2 \times [0, R]$. The spatial arrangement of the nodes is described by a 3D point process, $\Phi_{3D} = \{ (x_i, y_i, z_i) \}$, which is constructed as follows (illustrated in Fig.~\ref{fig:SystemModel4Net}). The projection of the node locations onto the horizontal plane, denoted by $\Phi_{2D} = \{ (x_i, y_i) \}$, constitutes a homogeneous Poisson Point Process (PPP) on $\mathbb{R}^2$ with a constant spatial intensity of $\Lambda > 0$. Furthermore, each node is assigned a height, $z_i$, which is an independent and identically distributed (i.i.d.) random variable drawn from a uniform distribution, $\mathcal{U}(0, R)$.
        
        The intensity of this 3D process, $\mathcal{N}$, can be derived from these properties. For an infinitesimal volume element $\mathrm{d}V = \mathrm{d}A \times \mathrm{d}z$ located within the slab $S$, the expected number of nodes from the 2D projected process within the area $\mathrm{d}A$ is $\mathbb{E}[N_{2D}(\mathrm{d}A)] = \Lambda \, \mathrm{d}A$ \cite{streit2010poisson}. As the height of each node is uniformly and independently distributed, the probability of a node falling within the infinitesimal height interval $dz$ is given by $P(z_i \in \mathrm{d}z) = \mathrm{d}z/R$. Leveraging the independence of the planar process and the height distribution, the expected number of 3D nodes within the volume $\mathrm{d}V$ is:
        \begin{align}
            \mathbb{E}[N_{3D}(\mathrm{d}V)] &= \mathbb{E}[N_{2D}(\mathrm{d}A)] \cdot P(z_i \in \mathrm{d}z) \nonumber \\ 
            &= (\Lambda \, \mathrm{d}A) ( \frac{\mathrm{d}z}{R} ) = \frac{\Lambda}{R} \, \mathrm{d}V .
        \end{align}
        This confirms that $\Phi_{3D}$ is a homogeneous PPP within the slab $S$, characterized by a constant 3D intensity of $\mathcal{N} = \Lambda/R$.
        
        Without loss of generality, we model a communication link between any two nodes, as depicted in Fig.~\ref{fig:SystemModel4Pair}. The transmitter (Tx) is positioned at the origin of the coordinate system, $O(0,0,0)$. The transmitter, which may consist of a single Light Emitting Diode (LED) or an array thereof, is oriented to direct its emission towards the seafloor along the positive z-axis\footnote{If it is not pointed directly at the seafloor, the analysis can be continued by adding a bias. Furthermore, in scenarios where the transmitter experiences random orientation changes due to hydrodynamics, the coordinate system can be mathematically realigned to the instantaneous optical axis. Due to the isotropy of the homogeneous Poisson point process, the statistical properties of the receiver distribution remain invariant under this rotation. Consequently, fixing the transmitter orientation serves as a valid canonical model that captures the relative geometric randomness without loss of generality.}. The transmitter's emission pattern is modeled as Lambertian, which creates a conical coverage volume defined by a characteristic semi-angle at half power, $\phi_{1/2}$.
        
        An optical receiver (Rx) is randomly located at point $P$, which is described by Cartesian coordinates $(x, y, z)$ or, equivalently, cylindrical coordinates $(r, \phi, z)$. The angle of irradiance at the Tx is denoted by $\theta$, while the angle of incidence at the Rx is $\psi$. The direct Euclidean distance between the Tx and the Rx is given by $L_{dis} = \sqrt{x^2 + y^2 + z^2}$. The depth difference between the Tx and Rx is denoted as $Z_0$.

    \subsection{Channel Model}
        %%%%%%%%%%%%%%%%%%%%%%%%%%%%%%%%%%%%%%%%%%%%%%%%%%%%%%%%%%%%%%%%%
    \begin{figure}[!t]
    \centering
    \includegraphics[width=0.8\columnwidth]{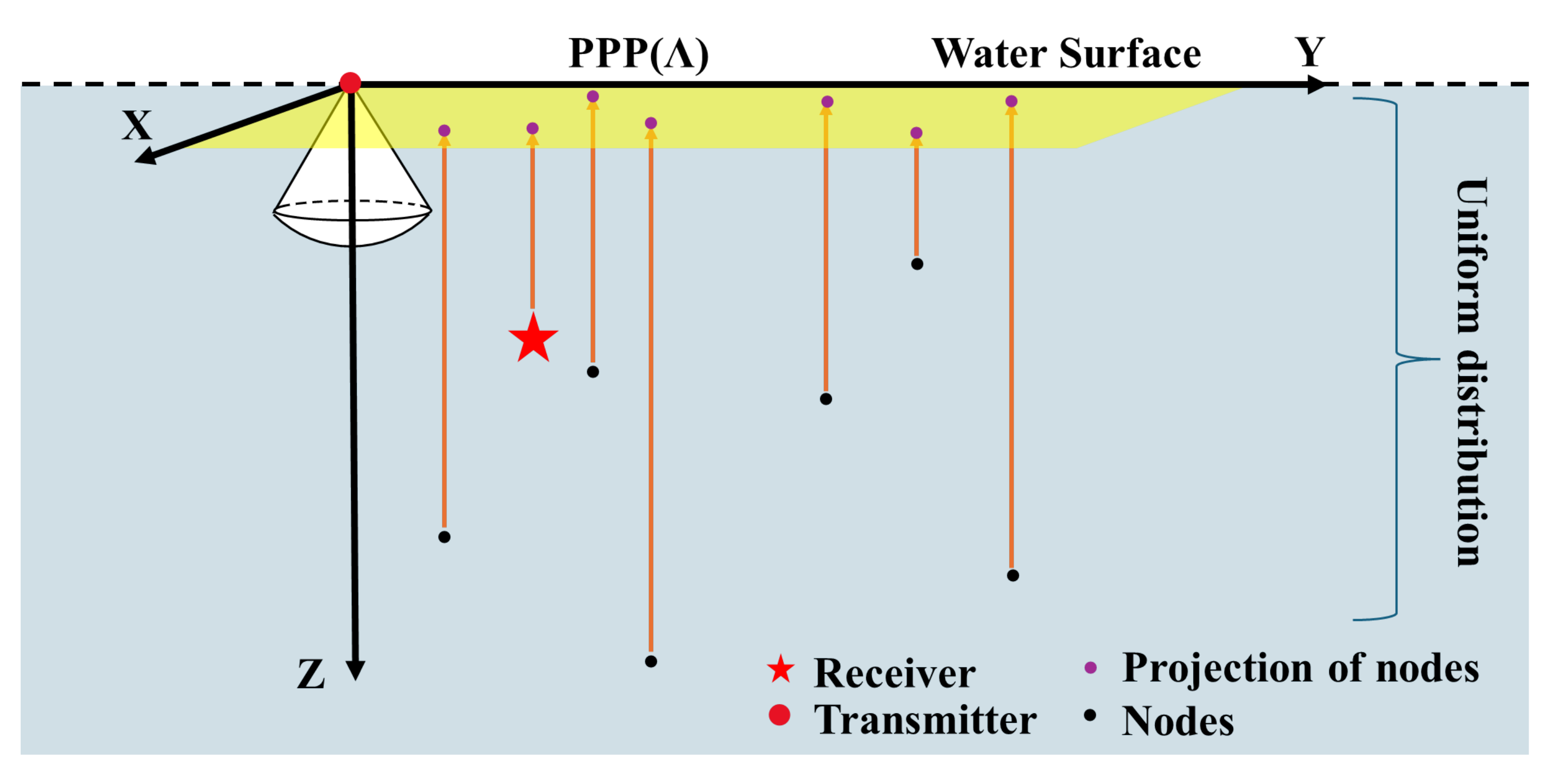}
    \caption{Model of actual underwater optical communication}
    \label{fig:SystemModel4Net}
\end{figure}
    %%%%%%%%%%%%%%%%%%%%%%%%%%%%%%%%%%%%%%%%%%%%%%%%%%%%%%%%%%%%%%%%%    %%%%%%%%%%%%%%%%%%%%%%%%%%%%%%%%%%%%%%%%%%%%%%%%%%%%%%%%%%%%%%%%%
\begin{figure}[!t]
    \centering
    \includegraphics[width=0.8\columnwidth]{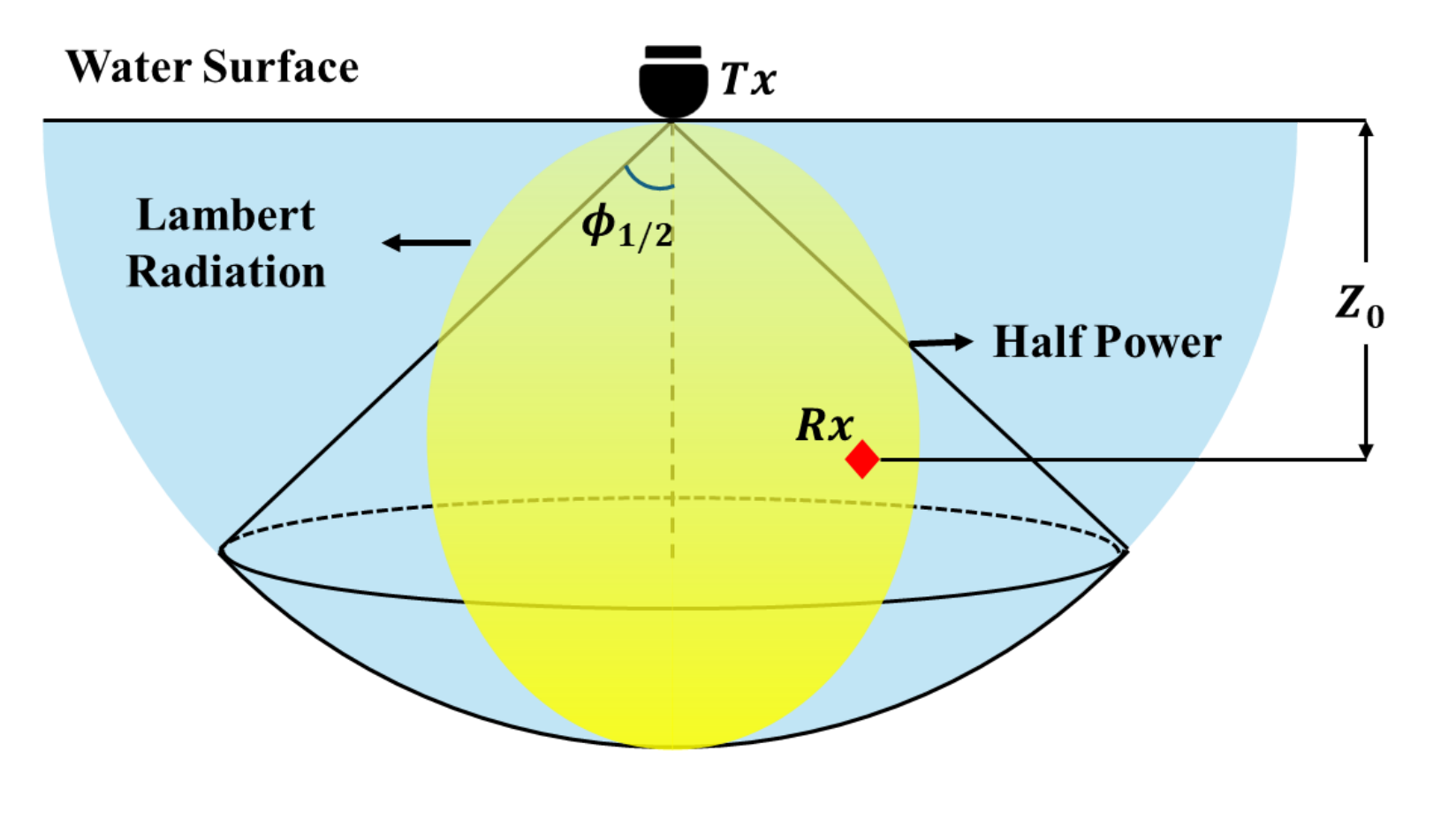}
    \caption{Mathematical model of a pair of Tx and Rx in underwater optical communication}
    \label{fig:SystemModel4Pair}
\end{figure}

    %%%%%%%%%%%%%%%%%%%%%%%%%%%%%%%%%%%%%%%%%%%%%%%%%%%%%%%%%%%%%%%%%
        \subsubsection{Extinction Coefficient}
            The propagation of an optical signal through seawater is primarily attenuated by absorption and scattering. These phenomena arise from the complex composition of the marine medium, which includes dissolved salts, suspended particulate matter, and chlorophyll from phytoplankton. The magnitudes of both absorption and scattering are functions of the signal's wavelength, $\lambda$, and the water depth, $z$.
            
            The cumulative impact of these processes is quantified by the total extinction coefficient, $c(\lambda, z)$, which is defined as the sum of the absorption coefficient, $a(\lambda, z)$, and the scattering coefficient, $b(\lambda, z)$ \cite{haltrin1999chlorophyll}:
            \begin{align}
            c(\lambda, z) = a(\lambda, z) + b(\lambda, z).
            \end{align}
            As established in the literature, the coefficients $a(\lambda, z)$ and $b(\lambda, z)$ can be decomposed into constituent components attributable to pure water, chlorophyll, and suspended sediments. This decomposition enables a precise estimation of the total extinction coefficient for specific environmental conditions.

            It is worth mentioning that the extinction coefficient $c(\lambda, z)$ exhibits complex dependencies on depth, salinity, and temperature gradients. Recognizing that the underwater optical channel serves as the fundamental physical foundation for signal attenuation and link reliability, its precise characterization is paramount to any rigorous performance analysis. However, to maintain the mathematical tractability of the stochastic geometry derivation, this study adopts a homogeneous medium assumption, using an effective average value for $c(\lambda)$ \cite{vaiopoulos2023uowc}. This model is primarily applicable to well-mixed deep-sea layers or open ocean environments. While it does not capture sharp hydrological transitions (e.g., haloclines or river plumes), it suffices to reveal the fundamental geometric laws governing the offset pointing strategy.
        \subsubsection{Line-of-Sight (LOS) Channel Gain}
            For a LOS link, the channel DC gain, $\mathcal{H}$, encapsulates the combined effects of the transmitter's emission pattern, the receiver's collection efficiency, and the attenuation along the propagation path. Accordingly, it is formulated as the product of three components: the transmitter's radiant intensity distribution ($\mathfrak{S}$), the path loss from medium extinction ($\mathcal{L}$), and the receiver's response ($\mathfrak{R}$).
            The transmitted power adheres to a generalized Lambertian radiation pattern. The radiant intensity, $\mathfrak{S}$, which quantifies the power per unit solid angle in a specific direction, is scaled by the squared distance to yield the irradiance at the receiver \cite{kahn1997wireless}:
            \begin{align}
            \mathfrak{S} = \frac{(m + 1)\cos^m(\theta)}{2\pi L_{dis}^2},
            \end{align}
            where $\theta$ is the angle of irradiance relative to the transmitter's primary axis. The directivity of the transmitter's beam is described by the Lambertian order ($m$), which can be calculated from the half-power semi-angle ($\phi_{1/2}$), using the following relationship:
            \begin{align}
            m = -\frac{\ln 2}{\ln(\cos(\phi_{1/2}))}.
            \end{align}

            The path loss, $\mathcal{L}$, is governed by the Beer-Lambert law, which models the exponential decay of optical power over the distance $L_{dis}$ due to absorption and scattering:
            \begin{align}
            \mathcal{L} = e^{-c(\lambda,z)L_{dis}}.
            \end{align}
            
            The receiver's response, $\mathfrak{R}$, is determined by its physical and optical properties. It incorporates the transmission of the optical bandpass filter ($T_s(\psi)$), the gain of the optical concentrator ($g(\psi)$), and the projection effect of the incident angle ($\cos(\psi)$) \cite{hamza2016investigation}. We can have
            \begin{align}
                    \mathfrak{R} = T_s(\psi) \cdot g(\psi) \cdot \cos(\psi).
            \end{align}
            The optical concentrator serves to enhance the signal collection capability. For an ideal non-imaging concentrator, the gain is given by:
            \begin{align}
            g(\psi) = 
            \begin{cases}
                \frac{n^2}{\sin^2(\phi_{FoV})}, & 0 \le \psi \le \phi_{FoV}  \\
                0, & \psi > \phi_{FoV}
            \end{cases},
            \end{align}
            where $n$ is the refractive index of the receiver's optical interface and $\phi_{FoV}$ is the receiver's field of view semi-angle. 
            
            Combining these components yields the total LOS channel gain:
            \begin{align}
            \mathcal{H} &= \mathfrak{S} \cdot \mathcal{L} \cdot \mathfrak{R} \nonumber\\ 
            &= \frac{(m+1)}{2\pi L_{dis}^2} \cos^m(\theta) \cos(\psi) T_s(\psi) g(\psi) e^{-c(\lambda,z)L_{dis}}.
            \label{eq:FullChannelGain}
            \end{align}
            The received power, $P_{Rx}$, is subsequently calculated as the product of the transmitted power, $P_{Tx}$, and this channel gain: $P_{Rx} = P_{Tx} \cdot \mathcal{H}$.

            As one of the most important KPIs for evaluating link quality, SNR can be expressed as:
            \begin{align}
                \text{SNR} = \frac{[ I_{signal}(P_{Tx}, L_{dis}, \delta, \phi_{1/2}) ]^2}{\sigma_{total}^2},
            \label{eq:snr_functional}
            \end{align}
where $\sigma_{total}$ is the sum of all the variances of the noise.

\section{Performance Evaluation}\label{sec:PE}

\subsection{Nearest Neighbor Distribution}
Based on existing research \cite{talgat2020nearest,yuan2020connectivity}, we can assume without loss of generality that a transmitter (Tx) forms a communication link with its nearest neighbor. We analyze the distance distribution from the origin to this nearest node. For a uniform PPP with intensity $\lambda$ in the region $\mathcal{X}$, the survival function is given by the following formula.

\begin{proposition}[Survival Function and PDF of Nearest Neighbor Distance]
    For a PPP of intensity $\lambda = \Lambda/R$ in a slab geometry $\mathcal{X} = \mathbb{R}^2 \times [0, R]$, the survival function $S(s) = \mathbb{P}(L_{dis} > s)$ and the corresponding CDF $F_{L_{dis}}(s)$ of the nearest neighbor distance $L_{dis}$ from the origin are given by:
    \begin{align}
        S(s) &= 
        \begin{cases}
            \exp(-\dfrac{2\pi\Lambda}{3R}s^3), & \!\!\!\!\!\!\!\!\!\!\!0 \le s \le R \\
            \exp(-\pi\Lambda(s^2 - \dfrac{R^2}{3})), & s > R
        \end{cases} ,
\end{align}
    \begin{align}
        F_{L_{dis}}(s) &=\!
        \begin{cases}
            1 - \exp(-\dfrac{2\pi\Lambda}{3R}s^3), & 0 \le s \le R \\
            1 - \exp(-\pi\Lambda(s^2 - \dfrac{R^2}{3})), & s > R
        \end{cases}.
\end{align}
    The PDF of $L_{dis}$ is:
    \begin{align}
       &\frac{\mathrm{d}F_{L_{dis}}(s)}{ds} =  f_{L_{dis}}(s) \nonumber
       \\  
        = &\begin{cases}
            \dfrac{2\pi\Lambda s^2}{R} \exp(-\dfrac{2\pi\Lambda s^3}{3R}), & 0 \le s \le R \\  
            2\pi\Lambda s \exp(-\pi\Lambda(s^2 - \dfrac{R^2}{3})), & s > R
        \label{eq:15}
        \end{cases}.
    \end{align}
\end{proposition}

Fig. \ref{Error2} demonstrates exact correspondence between the derived closed-form solution and Monte Carlo integration results, confirming analytical accuracy.
\begin{proof}
The survival function is defined as $S(s) = \exp(-\lambda \cdot \text{Vol}(B(\mathcal{O}, s) \cap \mathcal{X}))$ \cite{Baddeley2007}, where $B(\mathcal{O}, s)$ is a ball of radius $s$ centered at the origin $\mathcal{O}$. The effective intersection volume $V_{\text{eff}}(s)$ is calculated for two cases.

\textbf{Case 1: $0 \le s \le R$}. The intersection is a hemisphere within the slab, with volume:
\begin{align}
    V_{\text{eff}}(s) = \frac{2}{3}\pi s^3.
\end{align}

\textbf{Case 2: $s > R$}. The intersection is a spherical cap. The volume is found by integrating the circular cross-sections from $z=0$ to $z=R$:
\begin{align}
    V_{\text{eff}}(s) &= \int_0^R \pi(s^2 - z^2) \mathrm{d}z 
    = \pi R s^2 - \frac{\pi R^3}{3} .
\end{align}
Combining both cases gives the piecewise function for the effective volume:
\begin{align}
    V_{\text{eff}}(s) = 
    \begin{cases} 
        \dfrac{2}{3}\pi s^3, & 0 \le s \le R \\
        \pi R s^2 - \dfrac{\pi R^3}{3}, & s > R 
    \end{cases}.
\end{align}
Substituting $V_{\text{eff}}(s)$ and $\mathcal{N} = \Lambda/R$ into the survival function formula yields $S(s)$. The CDF is $F_{L_{dis}}(s) = 1 - S(s)$. The PDF is obtained by differentiating the CDF with respect to $s$ for each piece, which yields the stated result.
\end{proof}
    %%%%%%%%%%%%%%%%%%%%%%%%%%%%%%%%%%%%%%%%%%%%%%%%%%%%%%%%%%%%%%%%%

\begin{figure}[t]
    \centering
    \includegraphics[width=0.5\textwidth]{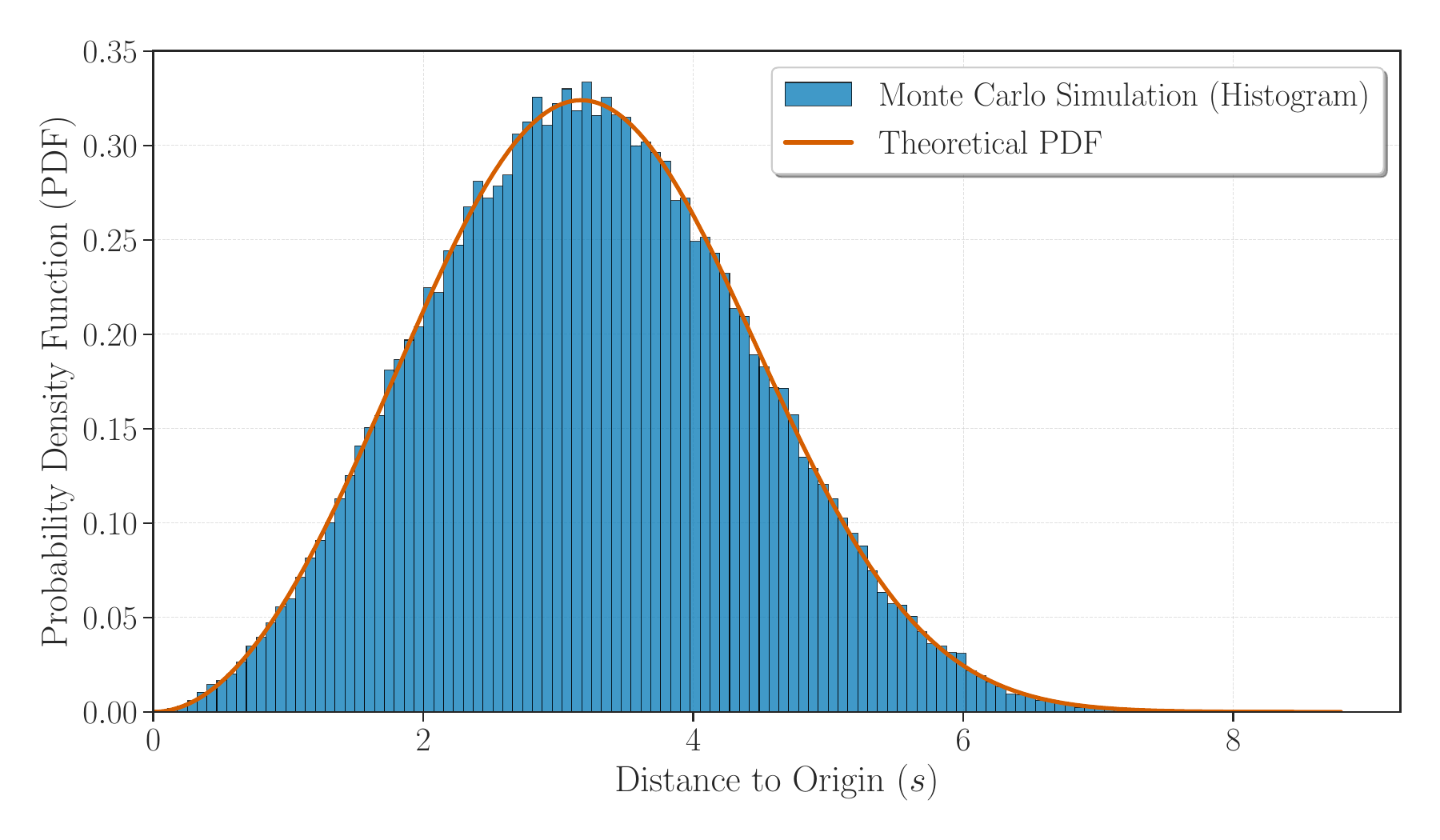}
    \caption{Comparison between closed-form solution and numerical solution for \eqref{eq:15}}
    \label{Error2}
\end{figure}
    %%%%%%%%%%%%%%%%%%%%%%%%%%%%%%%%%%%%%%%%%%%%%%%%%%%%%%%%%%%%%%%%%
\subsection{Expectation of Link Depth}
\begin{theorem}[Expected Vertical Position of Nearest Neighbor]\label{theorm32}
The expected vertical position $E[Z_0]$ of the nearest neighbor is given by:
\begin{align}
    E[Z_0] = \frac{1}{2} \int_{0}^{R} \exp(-\frac{2\pi\Lambda s^3}{3R}) \mathrm{d}s.
\end{align}
\end{theorem}

\begin{proof}
We compute the expected vertical position using the law of total expectation:
\begin{align}
    E[Z_0] = \int_{0}^{\infty} E[Z_0 \mid L_{dis} = s] f_{L_{dis}}(s) \mathrm{d}s.
\end{align}
The conditional expectation is given by two regimes based on the geometry:
\begin{subequations}
\begin{align}
    E[Z_0 \mid L_{dis} = s] &= \frac{s}{2}, \quad 0 \le s \le R ,\\
    E[Z_0 \mid L_{dis} = s] &= \frac{R}{2}, \quad s > R \quad.
\end{align}
\end{subequations}
We split the integral at $s=R$:
\begin{align}
    E[Z_0] &= \underbrace{\int_{0}^{R} \frac{s}{2} f_{L_{dis}}(s) \mathrm{d}s}_{I_1} + \underbrace{\int_{R}^{\infty} \frac{R}{2} f_{L_{dis}}(s) \mathrm{d}s}_{I_2}.
\end{align}
First, we evaluate $I_2$, which is $\frac{R}{2}$ times the probability that $L_{dis} > R$:
\begin{align}
    I_2 = \frac{R}{2} \int_{R}^{\infty} f_{L_{dis}}(s) \mathrm{d}s = \frac{R}{2} S(R) = \frac{R}{2} e^{(-\frac{2\pi\Lambda R^2}{3})}.
\end{align}
Next, we evaluate $I_1$. A direct approach is to use integration by parts on the definition of $I_1$. We can relate $I_1$, leveraging the relationship between the PDF $f(s)$ and the survival functions $S(s)$. Let $u = s/2$ and $\mathrm{d}v = f(s)\mathrm{d}s$. Then $\mathrm{d}u = (1/2)\mathrm{d}s$ and $v = F(s) = 1-S(s)$. Applying integration by parts to the expression $\int_0^R (s/2)f(s)\mathrm{d}s$
\begin{align}
I_1 &= \int_{0}^{R} \frac{s}{2} f(s) \mathrm{d}s \nonumber\\
&= \frac{R}{2}F(R) - \frac{1}{2}\int_0^R (1-S(s))\mathrm{d}s \nonumber\\
&= \frac{R}{2}(1 - S(R)) - \frac{R}{2} + \frac{1}{2}\int_0^R S(s)\mathrm{d}s \nonumber\\
&= -\frac{R}{2}S(R) + \frac{1}{2}\int_0^R S(s)\mathrm{d}s.
\end{align}
Combining the results for $I_1$ and $I_2$:
\begin{align}
    E[Z_0] = ( -\frac{R}{2}S(R) + \frac{1}{2}\int_0^R S(s)\mathrm{d}s ) + \frac{R}{2}S(R).
\end{align}
This simplifies to the final expression:
\begin{align}
    E[Z_0] = \frac{1}{2} \int_{0}^{R} S(s) \mathrm{d}s = \frac{1}{2} \int_{0}^{R} \exp(-\frac{2\pi\Lambda s^3}{3R}) \mathrm{d}s.
\end{align}
This integral can be evaluated numerically or expressed using the incomplete gamma function.
\end{proof}

\subsection{Received Power}
\label{sec:Int}
    %%%%%%%%%%%%%%%%%%%%%%%%%%%%%%%%%%%%%%%%%%%%%%%%%%%%%%%%%%%%%%%%%
\begin{figure*}[!t]
    \centering
    \subfloat[Geometric models and relationships between Tx and Rx]{
        \includegraphics[width=0.5\textwidth]{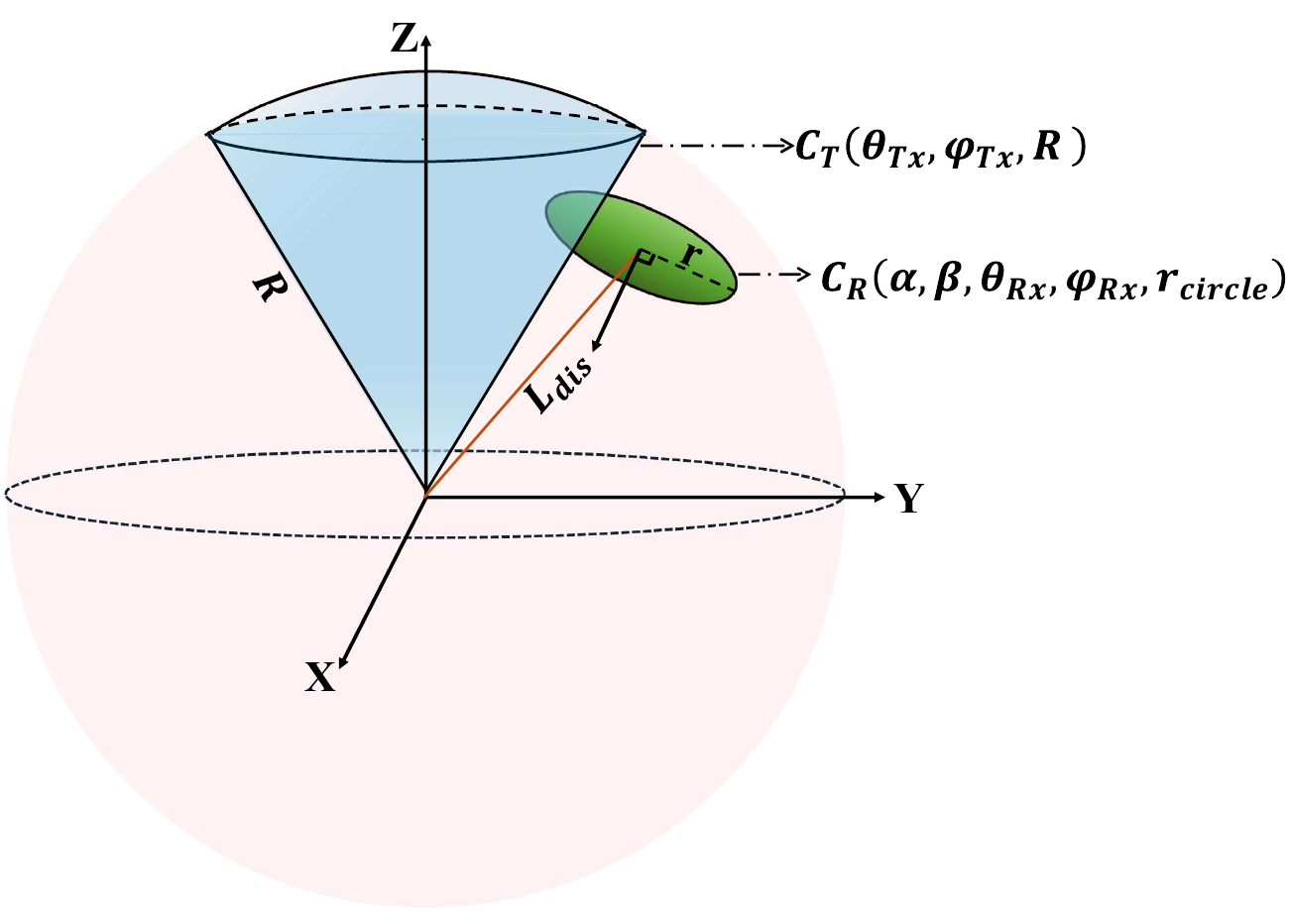}
        \label{Model For Tx and Rx}
    }
    \hfill
    \subfloat[Geometric model and parameter definition of Rx]{
        \includegraphics[width=0.4\textwidth]{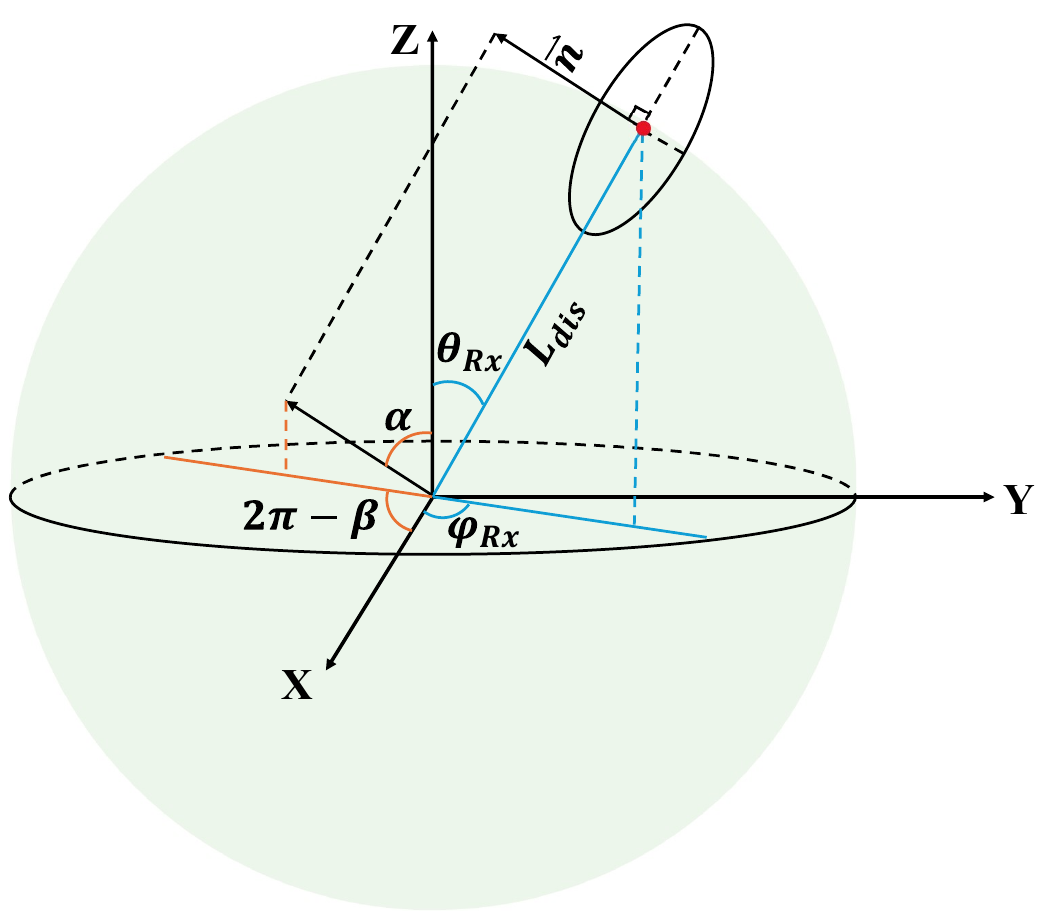}
        \label{Model for Rx}
    }
    \caption{Geometric models, mathematical definitions, and parameter representations related to Tx and Rx}
    \label{fig:angular_distributions}
        \vspace{-3mm}
\end{figure*}
    %%%%%%%%%%%%%%%%%%%%%%%%%%%%%%%%%%%%%%%%%%%%%%%%%%%%%%%%%%%%%%%%%
To model the optical wireless channel, we first establish a detailed geometric framework, as illustrated in Fig.~\subref*{Model For Tx and Rx}. In this configuration, the transmitter (Tx) is positioned at the origin of the coordinate system and is modeled as a radiating cone, $C_T$, characterized by the angles $\theta_{Tx}$ and $\phi_{Tx}$. The receiver (Rx) is modeled as a circular surface, $C_R$, with a radius $r_{circle}$. Its spatial placement and orientation are defined by the distance $L_{dis}$ separating its center from the origin, along with the orientation angles $\alpha$ and $\beta$. The specific geometric parameters relevant to the receiver plane, which are crucial for the power calculation, are further detailed in Fig.~\subref*{Model for Rx}.

The total received power, $P_{Rx}$, is subsequently found by integrating the incident optical power density over the entire receiver surface area. For this analysis, we assume the transmitter behaves as a Lambertian radiator of order $m$. This assumption governs the spatial distribution of the emitted radiation. Consequently, the differential received power, $\mathrm{d}P_{Rx}$, incident upon a differential element of the receiver's surface, is given by:
\begin{align}
\begin{aligned}
    \mathrm{d}P_{Rx} \!=\! \frac{1}{2} \mathfrak{R}P_{Tx}\frac{(m+1)}{2\pi}\! \frac{\cos^m(\theta)}{r^2} e^{-c(\lambda,L)r} |\cos\psi| \, \mathrm{d}A.
\end{aligned}
\end{align}

\begin{remark}
The $\frac{1}{2}$ term and the absolute value on $\cos\psi$ jointly account for the effect that light is received by only one side of the receiver.
\end{remark}
where $\mathfrak{R}$ is receiver responsivity, $c(\lambda,L)$ is the attenuation coefficient, $\gamma$ is the angle of incidence, and $dA = \rho \, d\rho \, d\theta'$ is the differential area element on the receiver disk. To find the average received power under random orientations, we integrate over all possible receiver positions $(\theta_{Rx}, \phi_{Rx})$ and normal vector orientations $(\alpha, \beta)$, assuming uniform distributions. This leads to the comprehensive integral:
\begin{align}
    P_{Rx} &= \frac{1}{2}\mathfrak{R}P_{Tx}\frac{m+1}{8 \pi^3} \int_0^{2\pi}\int_0^{\pi/2}\int_0^{2\pi}\int_{0}^{\pi}\int_0^{2\pi}\int_0^{r_{circle}} \nonumber \\
    &\quad \frac{\cos^m(\theta)}{r^2} e^{-c(\lambda,L)r} | \cos\gamma | \sin\theta_{Rx}\sin\alpha \, \rho \nonumber \\ 
    &\times \mathrm{d}\rho \, \mathrm{d}\theta' \, \mathrm{d}\alpha \, \mathrm{d}\beta \, \mathrm{d}\theta_{Rx} \, \mathrm{d}\phi_{Rx},
    \label{revi}
\end{align}
where $| \cos\gamma | = | \sin\theta_{Rx}\sin\alpha\cos(\phi_{Rx}-\beta)+\cos\theta_{Rx}\cos\alpha |$.

\begin{lemma}[Approximation for Large Distance]
For practical scenarios where the link distance is much larger than the receiver size ($L_{dis} \gg \rho$), the received power integral can be approximated as:
\begin{align}
    P_{Rx} &\approx C_0' \int_0^{2\pi}\int_0^{\pi/2}\int_0^{2\pi}\int_{0}^{\pi} (\cos\theta_{Rx})^m \sin\theta_{Rx}\sin\alpha \nonumber \\ 
    & \times | \sin\theta_{Rx}\sin\alpha\cos(\phi_{Rx}-\beta)+\cos\theta_{Rx}\cos\alpha | \nonumber
    \\&\times \mathrm{d}\alpha \, \mathrm{d}\beta \, \mathrm{d}\theta_{Rx} \, \mathrm{d}\phi_{Rx},
\end{align}
where $C_0' = \frac{\mathfrak{R}P_{Tx}(m+1)(\pi r_{circle}^2)e^{-c(\lambda,L)L_{dis}}}{16 \pi^3 L_{dis}^2 }$.
\end{lemma}

\begin{proof}
The approximation is based on the following simplifications for $L_{dis} \gg \rho$:
\begin{itemize}
    \item The distance to any point on the receiver is approximately the distance to its center: $r \approx L_{dis}$.
    \item The attenuation term is constant over the receiver surface: $e^{-c(\lambda,L)r} \approx e^{-c(\lambda,L)L_{dis}}$.
    \item The angle of emission is approximately the angle to the receiver's center: $\cos(\theta) \approx \cos\theta_{Rx}$.
\end{itemize}
Applying these simplifications, the terms depending on $r$ and $\theta$ can be moved outside the integral over the receiver surface ($\mathrm{d}A = \rho \mathrm{d}\rho \mathrm{d}\theta'$). The integral over the receiver surface then yields its area, $\int_0^{2\pi}\int_0^{r_{circle}} \rho \mathrm{d}\rho \mathrm{d}\theta' = \pi r_{circle}^2$. Substituting this area and the approximations into the general expression yields the result.
\end{proof}

\begin{theorem}[Average Received Power with Random Orientation]
When the receiver orientation and position are uniformly random over a hemisphere, the average received power, under the large distance approximation, is:
\begin{align}
    P_{Rx} \approx \frac{\mathfrak{R} P_{Tx} r_{circle}^2}{4 L_{dis}^2} e^{-c(\lambda,L)L_{dis}}
    \label{eq:40}.
\end{align}
\end{theorem}

\begin{remark}
In the stochastic framework, this aperture integration quantifies the expected optical flux under spatial uncertainty, rather than the deterministic coupling of a static link. Maximizing this metric is equivalent to maximizing the average captured power, serving as a robust reliability proxy for random deployments.
\end{remark}
\begin{proof}
Starting from the approximated integral in the Lemma, we perform the integration sequentially.
First, we substitute $\beta' = \phi_{Rx} - \beta$. Since the integration over $\phi_{Rx}$ is over a full period $[0, 2\pi]$, this step simplifies the dependency and the integral over $\phi_{Rx}$ yields a factor of $2\pi$. The integral becomes:
\begin{align}
&P_{Rx} \approx 2\pi C_0' \int_0^{\pi/2}\int_0^{2\pi}\int_{0}^{\pi} (\cos\theta_{Rx})^m \sin\theta_{Rx}\sin\alpha \nonumber\\
& \times | \!\sin\theta_{Rx}\!\sin\alpha\!\cos(\beta)\!+\ \!cos\theta_{Rx}\cos\alpha | \,\mathrm{d}\alpha  \mathrm{d}\beta  \mathrm{d}\theta_{Rx}.
\end{align}
The integral over $\beta$ from $0$ to $2\pi$ of the absolute value term is solved analytically, which results in a complex expression. As shown in the derivation, this leads to:
\begin{align}
P_{Rx} &\approx 2\pi C_0' \int_0^{\frac{\pi}{2}} \!\! \Bigg[\int_0^{\frac{\pi}{2} - \theta_{Rx}} \!\!\!\!\pi (\cos\theta_{Rx})^{m+1} \sin\theta_{Rx} \sin 2\alpha  \mathrm{d}\alpha \nonumber \\
& + \int_{\frac{\pi}{2} - \theta_{Rx}}^{\frac{\pi}{2} + \theta_{Rx}} (\cos\theta_{Rx})^m (\sin\theta_{Rx})^2 (\sin\alpha)^2 F(Y)  \mathrm{d}\alpha \nonumber \\
& + \int_{\frac{\pi}{2} + \theta_{Rx}}^{\pi} \!\!\!\!-\pi (\cos\theta_{Rx})^{m+1} \sin\theta_{Rx} \!\sin 2\alpha  d\alpha \!\Bigg]\!\mathrm{d}\theta_{Rx} ,
\end{align}
where
$Y = \frac{\cos\theta_{Rx} \cos\alpha}{\sin\theta_{Rx} \sin\alpha}$ and
$F(Y) = 4 \sqrt{1 - Y^2} + 2\pi Y - 4 Y \arccos(Y)$

This expression is then integrated ove $\alpha$. 
\begin{align}
P_{Rx} &\approx 2\pi C_0' \int_0^{\pi/2} \Bigg[ 2\pi (\cos\theta_{Rx})^{m+3} \sin\theta_{Rx}  \nonumber\\
&\quad+2 \pi \cos^m\theta_{Rx}\sin^2\frac{\theta_{Rx}}{2} \sin \theta_{Rx}\nonumber \\
&\quad \times (3+2\cos\theta_{Rx}+\cos2\theta_{Rx}) \Bigg] \mathrm{d}\theta_{Rx}.
\end{align}
Completing the final integration with respect to $\theta_{Rx}$ over $[0, \pi/2]$ yields:
\begin{align}
    P_{Rx} &\approx \frac{4 \pi^2 C_0'}{m+1} = \frac{\mathfrak{R} P_{Tx} r_{circle}^2}{4 L_{dis}^2} e^{-c(\lambda,L)L_{dis}}.
\end{align}
This confirms the stated result.
\end{proof}

\begin{corollary}[Power Received in Main Lobe]
If the receiver is located within the transmitter's main lobe, defined by the half-power angle $\phi_{1/2}$, the average received power is:
\begin{align}
    P_{Rx} = \frac{\mathfrak{R} P_{Tx} r_{circle}^2}{4 L_{dis}^2} e^{-c(\lambda,L)L_{dis}} ( 1 - \cos^{m+1}(\phi_{1/2}) )
    \label{eq:41},
\end{align}
where $m = -\ln(2)/\ln(\cos(\phi_{1/2}))$.
\end{corollary}

\begin{proof}
This result is obtained by following the proof of Theorem \ref{theorm32}, but changing the upper integration limit for $\theta_{Rx}$ from $\pi/2$ to $\phi_{1/2}$. The integration of $(\cos\theta_{Rx})^m \sin\theta_{Rx}$ terms leads directly to the factor $(1 - \cos^{m+1}(\phi_{1/2}))$.
\end{proof}

\begin{corollary}[Power with Pointing Error]
With a static pointing error offset angle $\delta$, the received power within the main lobe is modified to:
\begin{align}
\begin{aligned}
    P_{Rx} = &\frac{\mathfrak{R} P_{Tx}}{4} ( \frac{r_{circle}}{L_{dis}} )^2 e^{-c(\lambda,L)L_{dis}}\\
   & ( \cos^{m+1}(\delta) - \cos^{m+1}(\delta+\phi_{1/2}) ).
    \label{eq:42}
    \end{aligned}
\end{align}
This expression is valid for $\delta+\phi_{1/2} \leq \pi/2$.
\end{corollary}

\begin{proof}
The pointing error modifies the effective angle in the Lambertian term to $(\theta_{Rx}+\delta)$. The integration with respect to $\theta_{Rx}$ is performed over $[0, \phi_{1/2}]$, leading to the change in the cosine terms.
\end{proof}

\begin{proposition}[Optimal Pointing Offset]
For a given half-power angle $\phi_{1/2}$, the received power is maximized by selecting an optimal offset $\delta_{opt}$:
\begin{align}
    \delta_{opt} = \frac{\pi}{12} + \frac{1}{2\pi}(\cos(2\phi_{1/2}+\frac{4}{3}\pi)-1).
\end{align}
\end{proposition}

\begin{proof}
 The optimal offset $\delta_{opt}$ is found by numerically solving the equation $\sin(\delta)\cos^m(\delta) = \sin(\delta+\phi_{1/2})\cos^m(\delta+\phi_{1/2}) $ for various values of $\phi_{1/2}$. The formula presented here is a highly accurate empirical approximation to these numerical results. Fig. \ref{fig:Error3} shows the high correspondence between numerical and analytical results.
\end{proof}

Physically, the optimality of $\delta_{opt}$ arises from the trade-off between radiant intensity and spatial geometry. In the stochastic integral of expected power \eqref{revi}, the beam intensity peaks at the optical axis ($\theta_{Rx}=0$), whereas the spatial volume element, scaled by the Jacobian determinant $\sin(\theta_{Rx})$, is minimal at the axis and increases with the polar angle. Consequently, a vertical alignment inefficiently concentrates reception sensitivity on a region with negligible spatial probability. The offset strategy rotates the receiver to statistically align with the angular region where the product of photon flux and spatial volume is maximized, thereby optimizing the expected energy capture across the random node distribution.

To quantify the performance gap, we establish an ideal PAT scenario as the theoretical upper bound, denoted as $P_{Rx}^{PAT}$. In this deterministic regime, perfect knowledge of location and orientation allows the transmitter and receiver to maintain precise alignment. Consequently, the transmission and incidence angles are nullified ($\theta = \psi = 0$), and the stochastic orientation coefficient is removed. 

However, maintaining such perfect alignment is practically infeasible in dynamic underwater environments. Underwater nodes are subject to continuous hydrodynamic disturbances , and practical mechanical steering systems often exhibit residual pointing errors ranging from $\pm 10^\circ$ to $\pm 60^\circ$ depending on the device's stability and response time\cite{saeed2019underwater}.

Under far-field conditions ($L_{dis} \gg r_{circle}$), the received power can be directly derived by integrating the intensity over the receiver aperture $A_r$, yielding the standard LOS link budget with maximum axial gain,
\begin{equation}
    P_{Rx}^{PAT} (\epsilon) \approx \frac{\mathfrak{R} P_{Tx} (m+1) A_r}{2\pi L_{dis}^2} e^{-c(\lambda,L)L_{dis}} \cos^m(\epsilon),
    \label{eq:PAT_limit}
\end{equation}
where $\epsilon$ is pointing error. To contextualize the proposed strategy against this benchmark, Fig.~\ref{fig:PAT_comparison} presents a comparative analysis of normalized received power. The evaluation contrasts the active PAT system subject to $\epsilon$ against the passive stochastic strategies, specifically the baseline configuration ($\delta=0^\circ$) and the proposed offset strategy ($\delta=\delta_{opt}$).
\begin{figure}[!b]
    \centering
    \includegraphics[width=\columnwidth]{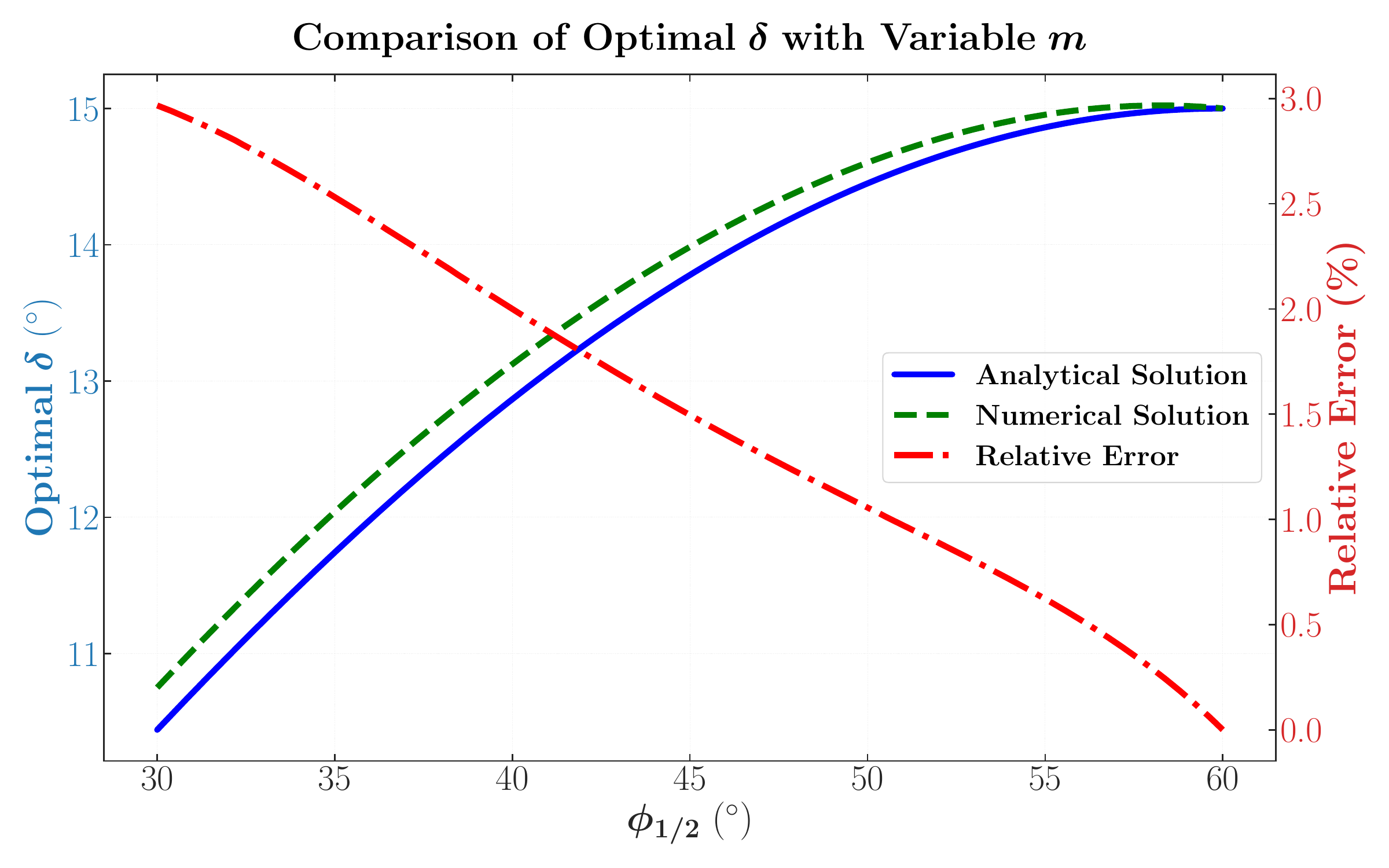}
    \caption{Comparison between closed-form solution and numerical solution for optimal $\delta$ in \eqref{eq:42}}
    \label{fig:Error3}
\end{figure}
% --- Figure Insertion ---
\begin{figure}
    \includegraphics[width=\columnwidth]{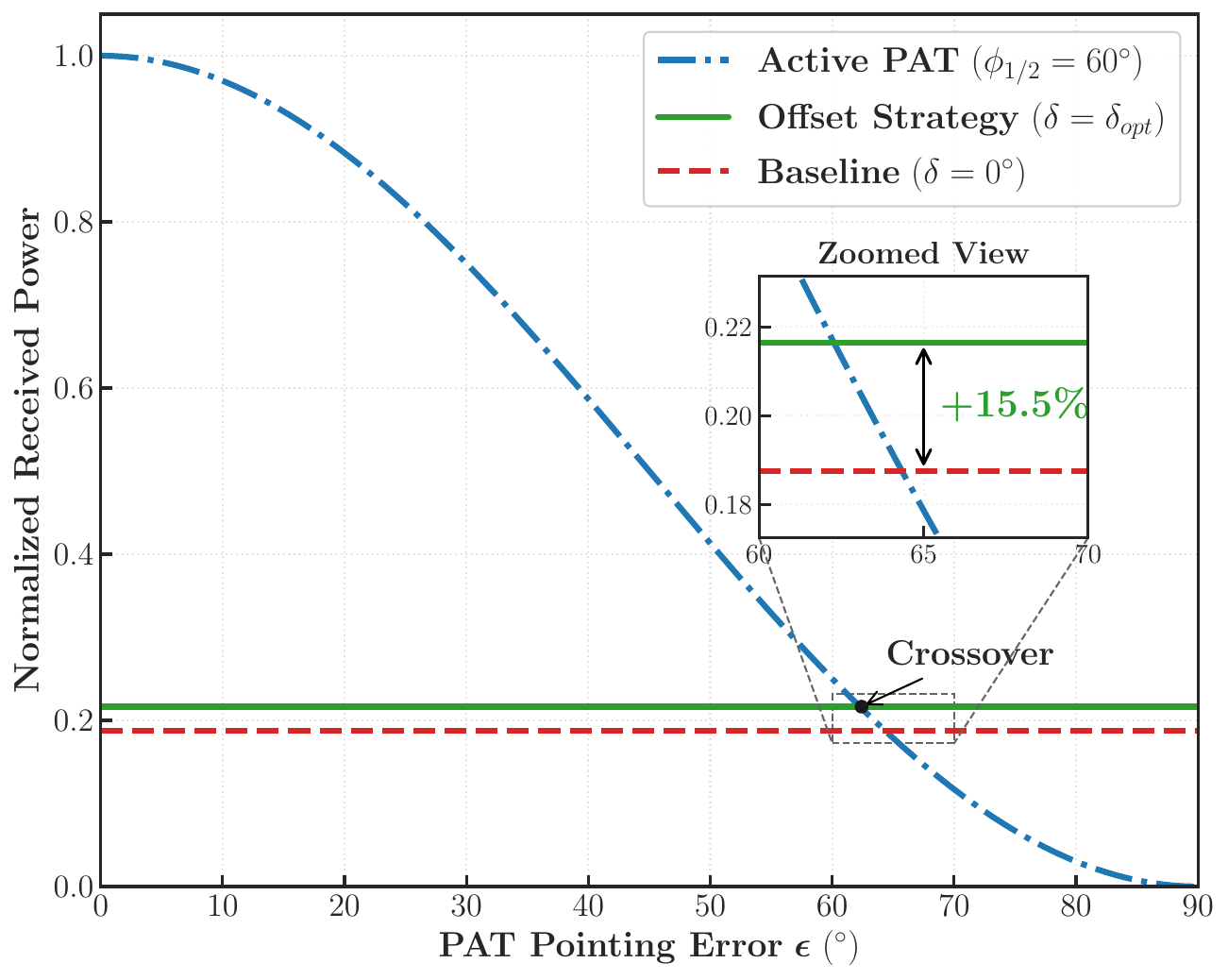} 
     \caption{Normalized received power comparison: active PAT with pointing error $\epsilon$ vs. passive stochastic strategies for $\phi_{1/2}=60^\circ$.}
    \label{fig:PAT_comparison}
\end{figure}
% ------------------------

The numerical results elucidate a fundamental trade-off between complexity and robustness. While the ideal PAT yields the theoretical maximum, the passive offset strategy establishes a stable performance floor without requiring complex servo-mechanisms. Notably, for broad-beam configurations ($\phi_{1/2}=60^\circ$), the curves exhibit a critical intersection point at approximately $62.5^\circ$. This indicates that the proposed passive offset strategy outperforms an active PAT system when the tracking error $\epsilon$ exceeds $62.5^\circ$, thereby offering a robust and cost-effective alternative for nodes with lower tracking precision.
    %%%%%%%%%%%%%%%%%%%%%%%%%%%%%%%%%%%%%%%%%%%%%%%%%%%%%%%%%%%%%%%%%

    %%%%%%%%%%%%%%%%%%%%%%%%%%%%%%%%%%%%%%%%%%%%%%%%%%%%%%%%%%%%%%%%%
    %%%%%%%%%%%%%%%%%%%%%%%%%%%%%%%%%%%%%%%%%%%%%%%%%%%%%%%%%%%%%%%%%
\begin{figure*}[htbp]
    \centering
    \subfloat[3D heat map of normalized $P_{Rx}$ under the influence of different $\delta$ and $\phi_{1/2}$]{
        \includegraphics[width=0.3\textwidth]{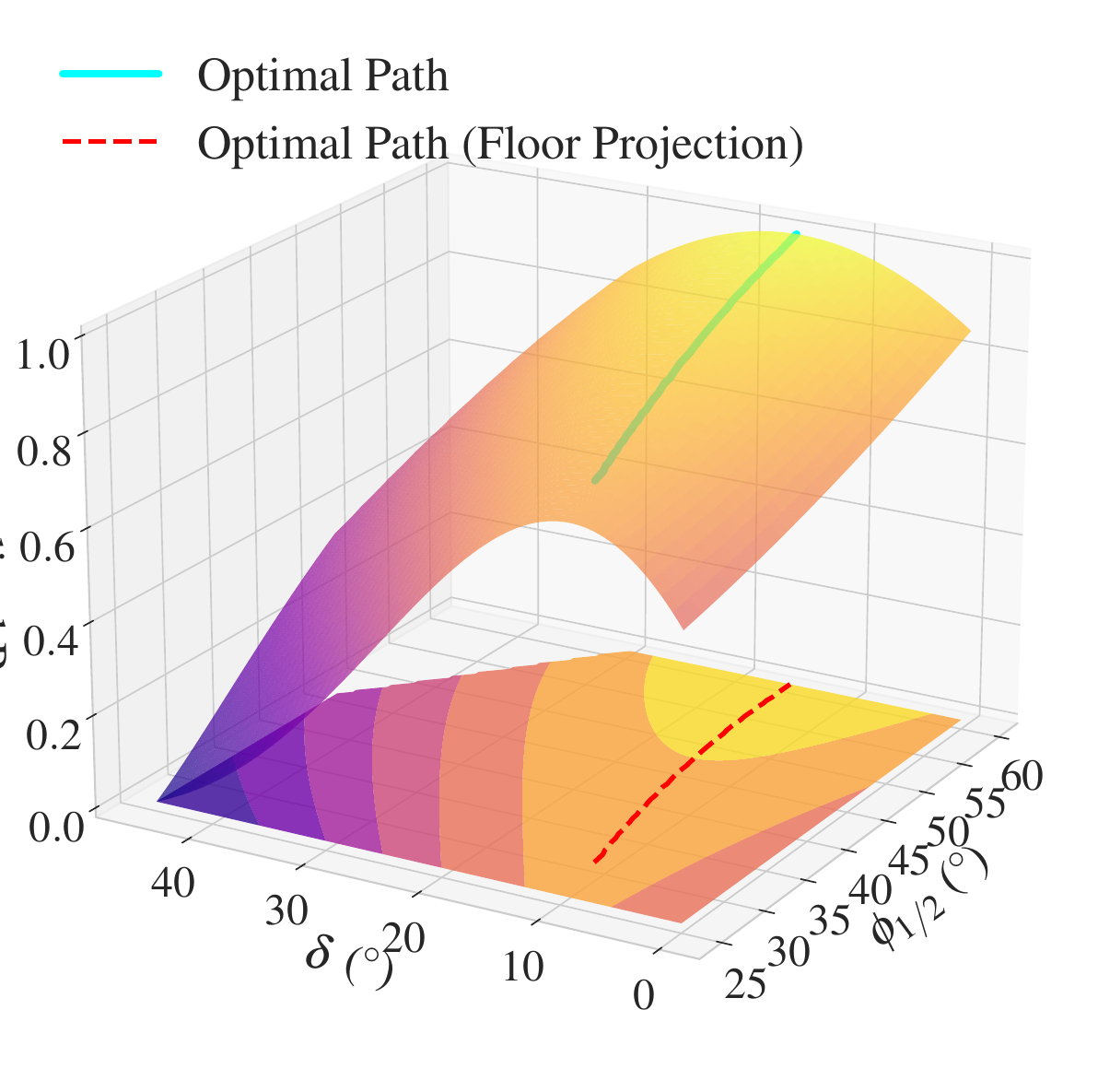}
        \label{fig:sub1}
    }
    \hfill
    \subfloat[2D heat map of normalized $P_{Rx}$ under the influence of different $\delta$ and $\phi_{1/2}$]{
        \includegraphics[width=0.3\textwidth]{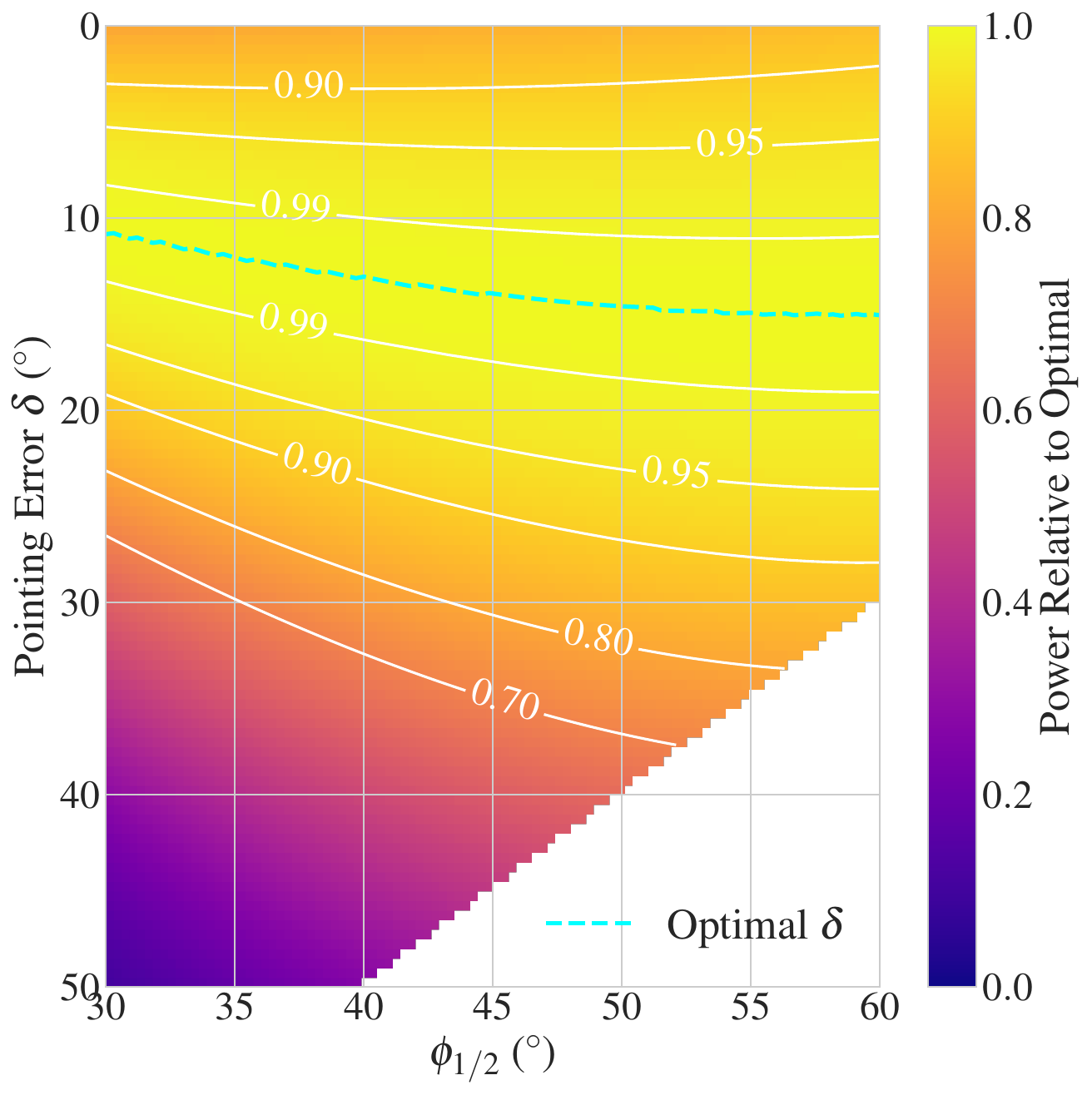}
        \label{fig:sub2}
    }
    \hfill
    \subfloat[Relationship curve between normalized $P_{Rx}$ and $\delta$ under typical $\phi_{1/2}$]{
        \includegraphics[width=0.3\textwidth]{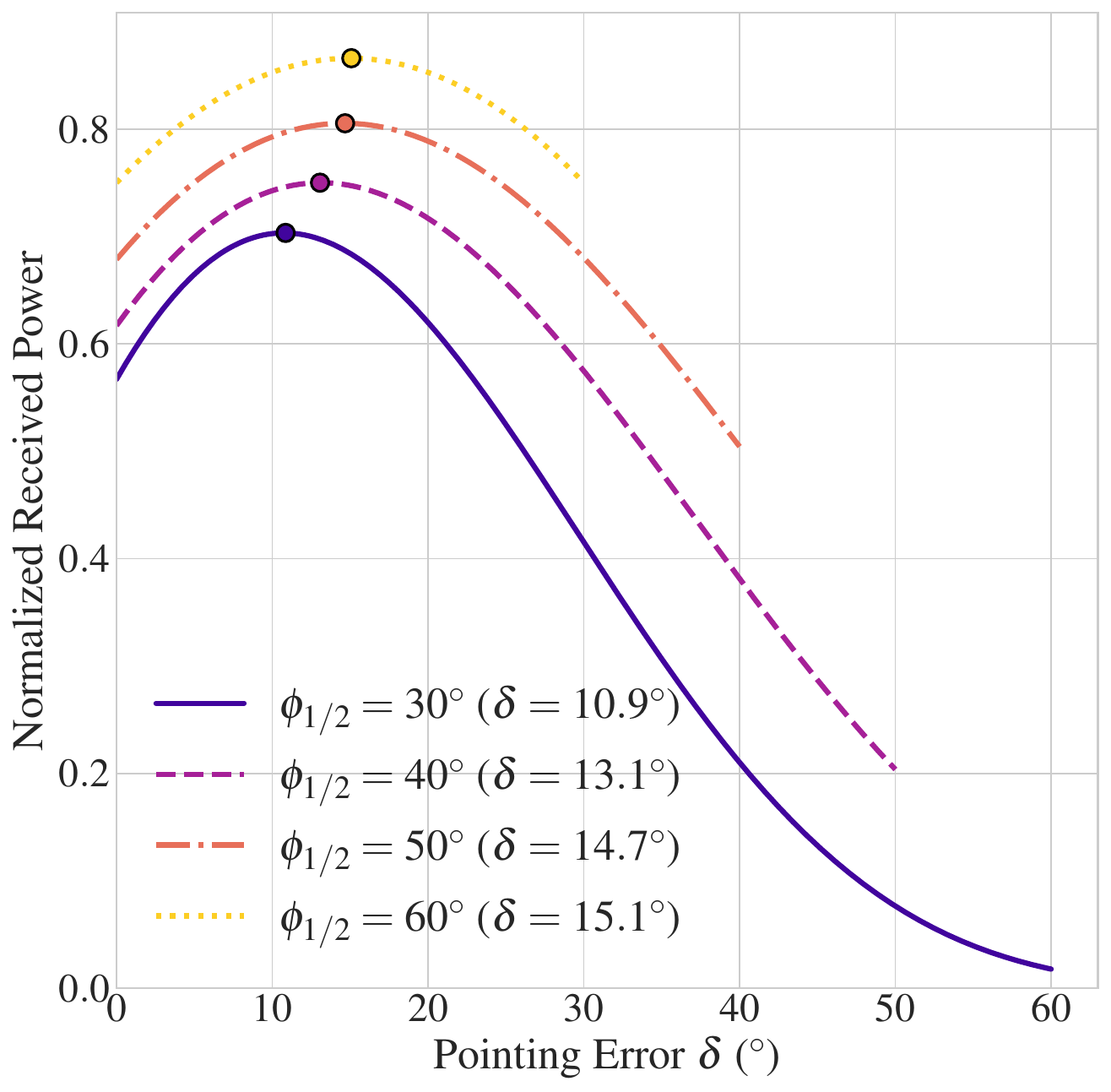}
        \label{fig:sub3}
    }
    \caption{Simulation results showing the effect of $\delta$ on $P_{Rx}$ under different performance forms}
    \label{fig:three_subfigs}
\end{figure*}
    %%%%%%%%%%%%%%%%%%%%%%%%%%%%%%%%%%%%%%%%%%%%%%%%%%%%%%%%%%%%%%%%%

    \begin{figure*}[!b]
 \hrulefill

 \begin{equation}
 \label{eq:fully_expanded_snr}
 \text{SNR} = \frac
 { ( \frac{\eta \lambda q}{hc} G P_{Rx} )^2  }
 {
     \begin{aligned}
         & 2qBG^2 ( \frac{1}{1 + \ln(1 - P_{\text{ct}})} )\!\!\!\!\!\!
         & \left[
             ( \frac{\eta \lambda q}{ hc}  P_{Rx} ) + I_d +
             ( \frac{\eta \lambda q}{hc}  ( \frac{\pi D^2}{4} \phi_{\text{FOV}}^2 L_f \zeta_r E_{\text{sun}}(z=0) e^{-\epsilon L_{\text{deep}}} \Delta\lambda T_s(\psi) ) )
         \right]
         + \frac{4 \kappa T B}{R_L}
     \end{aligned}
 }.
 \end{equation}
  
 \end{figure*}

\subsection{SNR and BER}
\label{sec:snr}
The significant attenuation of optical signals in the underwater environment makes the choice of photodetector critical. In this work, a Silicon Photomultiplier (SiPM) is employed as the receiver due to its distinct advantages for UOWC. A key feature of the SiPM is its high internal gain, typically on the order of $10^6$, which is a consequence of its architecture comprising an array of series-connected single-photon avalanche diodes (SPADs). This characteristic yields high sensitivity to weak signals, rendering the SiPM highly suitable for photon-starved UOWC links \cite{gundacker2020silicon}.

\begin{remark}
The high internal gain ($G$) of the SiPM allows it to generate a substantial photocurrent even from a low received optical power ($P_{Rx}$), which is fundamental for detecting signals after severe underwater attenuation.
\end{remark}

\noindent The resulting signal photocurrent, $I_{SiPM}$, is modeled as follows.

\begin{proposition}[Signal Photocurrent]
The signal photocurrent generated by the SiPM in response to a received optical power $P_{Rx}$ is given by:
\begin{align}
    I_{SiPM} = \frac{\eta\lambda q}{hc} G P_{Rx},
    \label{eq:sipm_photocurrent}
\end{align}
where $\eta$ is the photon detection efficiency (PDE), $\lambda$ is the wavelength, $q$ is the elementary charge, $h$ is Planck's constant, $c$ is the speed of light, and $G$ is the internal gain of the SiPM.
\end{proposition}

\begin{theorem}[Total Noise Variance]
The overall performance of the SiPM-based receiver is limited by several noise sources. The total noise variance, $\sigma_{total}^2$, is the linear superposition of four primary components: quantum shot noise ($\sigma_q^2$), dark current noise ($\sigma_d^2$), solar background noise ($\sigma_{\text{solar}}^2$), and thermal noise ($\sigma_{\text{th}}^2$).
\begin{align}
    \sigma_{total}^2 = \sigma_q^2 + \sigma_d^2 + \sigma_{\text{solar}}^2 + \sigma_{\text{th}}^2.
\end{align}
\end{theorem}

\begin{proof}
The total noise variance is derived by modeling each contributing source independently.
\begin{itemize}
    \item \textbf{Quantum Shot Noise}: This noise arises from the statistical fluctuations in photon arrival. Its variance is given by \cite{ghassemlooy2019optical}:
\begin{align}
        \sigma_q^2 = 2q I_{SiPM} B F G,
    \end{align}
    where $B$ is the receiver bandwidth and $F$ is the excess noise factor.
    
    \begin{remark}
    The excess noise factor, $F$, accounts for the additional noise introduced by the avalanche multiplication process, which is strongly influenced by secondary effects like optical crosstalk between SPAD cells. According to the model in \cite{vinogradov2012analytical}, it is related to the crosstalk probability, $P_{ct}$, by:
    \begin{align}
        F = \frac{1}{1+\ln(1-P_{ct})}.
    \end{align}
    \end{remark}
    
    \item \textbf{Dark Current Noise}: This intrinsic noise is generated by thermal excitation within the semiconductor material, even without illumination \cite{ghassemlooy2019optical}. Its variance is also subject to avalanche multiplication:
\begin{align}
        \sigma_d^2 = 2q I_d B F G,
    \end{align}
    where $I_d$ is the dark current.
    
    \item \textbf{Ambient Solar Noise}: In UOWC, a dominant noise contributor is the ambient background light from solar radiation. Following the methodology of \cite{1605919, 4686801}, we model the solar irradiance at depth $L_{deep}$ and the corresponding optical power $P_{sun}$ captured by the receiver:
    \begin{align}
        P_{sun} = A_r \phi_{FOV}^2 L_f \zeta_r E_{sun}(z=0) e^{-\epsilon L_{deep}} T_s(\psi),
    \end{align}
    where $A_r = \pi D^2/4$ is the receiver aperture area, $\phi_{FOV}$ is the half-angle field of view, $L_f$ is a directional dependence factor, $\zeta_r$ is the solar reflectance, $E_{sun}(z=0)$ is the surface irradiance (e.g., $1000 \, \text{W/m}^2$), $\epsilon$ is the solar radiation attenuation coefficient, and $T_s(\psi)$ is the optical system transmittance\footnote{Note that this formula offers a generalized model covering all depths. Through the attenuation term $e^{-\epsilon L_{deep}}$, the equation mathematically ensures that $P_{sun}$ goes to $0$ as depth increases. This guarantees that the model automatically reflects the physical reality where deep-sea scenarios (e.g., Level 3 in Section \ref{sec:SR}) are dominated solely by thermal and dark current noise.}  . This induces a photocurrent with a shot noise variance of:
    \begin{align}
        \sigma_{\text{solar}}^2 = 2q ( \frac{\eta\lambda q}{hc} G P_{sun} ) B F.
    \end{align}
    
    \item \textbf{Thermal Noise}: This noise, also known as Johnson-Nyquist noise, is generated by the electronic components, particularly the load resistance $R_L$. Its variance is independent of the optical signal:
    \begin{align}
        \sigma_{th}^2 = \frac{4\kappa TB}{R_{L}},
    \end{align}
    where $\kappa$ is the Boltzmann constant and $T$ is the absolute temperature.
\end{itemize}
The linear summation of these independent noise variances yields the total noise variance as stated in the theorem.
\end{proof}

\begin{theorem}[Overall Electrical SNR]
The overall electrical SNR is defined as the ratio of the squared signal photocurrent to the total noise variance:
\begin{align}
    \mathrm{SNR} = \frac{I_{SiPM}^2}{\sigma_{total}^2} = \frac{I_{SiPM}^2}{\sigma_q^2 + \sigma_d^2 + \sigma_{\text{solar}}^2 + \sigma_{th}^2}.
    \label{eq:snr_definition}
\end{align}
Substituting the expressions for the signal and noise components yields the comprehensive model presented in \eqref{eq:fully_expanded_snr}.
\end{theorem}
%%%%%%%%%%%%%%%%%%%%%%%%%%%%%%%%%%%%%%%%%%%%%%%%%%%%%%%%%%%%%%%%%

 %%%%%%%%%%%%%%%%%%%%%%%%%%%%%%%%%%%%%%%%%%%%%%%%%%%%%%%%%%%%%%%%%

\begin{corollary}[Bit Error Rate for NRZ-OOK]
Assuming the aggregate noise follows a Gaussian distribution, a reasonable approximation for a large number of detected photoelectrons, the theoretical BER for a non-return-to-zero on-off keying (NRZ-OOK) modulated signal is calculated from the SNR as:
\begin{align}
    \mathrm{BER} = \frac{1}{2} \mathrm{erfc}(\sqrt{\frac{\mathrm{SNR}}{2}}),
    \label{eq:ber}
\end{align}
where $\text{erfc}(\cdot)$ denotes the complementary error function.
\end{corollary}
\begin{proof}
This expression is the standard result for the BER of a baseband OOK signal with an optimal threshold detector under the assumption of additive white Gaussian noise (AWGN). The SNR in the argument of the erfc function directly links the physical layer characteristics encapsulated in \eqref{eq:fully_expanded_snr} to the system's communication performance.
\end{proof}
\subsection{Power Strategy}
\label{sec:power_strategy}
the primary objective of this section is to maximize the total number of transmittable bits within an energy-constrained UOWC network deployed over a circular area $\mathcal{A}$ with radius $R$. This is achieved by jointly optimizing the node density, $\Lambda$, and the transmit power, $P_{Tx}$. The analysis assumes OOK modulation, with energy consumption calculated for the transmission of bit `1' for simplicity.

To reflect practical engineering constraints where the total project budget (e.g., total battery mass or cost) is limited, we adopt a ``top-down"\footnote{Based on the total energy in certain area.} energy model with a fixed total budget, $E_{total}$, distinct from ``bottom-up''\footnote{Based on the energy of each node within a specific region and the total number of nodes.} approaches where total energy scales with node count. Consequently, a fundamental trade-off governs the network's performance. A higher node density $\Lambda$ reduces the average communication distance, thus potentially lowering the required $P_{Tx}$. However, it also divides the fixed total network energy, $E_{total}$, among a larger number of nodes, $N(\Lambda) = \mathcal{A} \cdot \Lambda$, thereby reducing the energy per node, $E_{node}$. Conversely, a lower density increases $E_{node}$ but necessitates a higher $P_{Tx}$ to overcome the increased path loss. This trade-off implies the existence of an optimal node density, $\Lambda^*$, that maximizes network throughput.Given the diverse definitions of network lifetime in underwater environments\cite{10.1145/1464420.1464425,5076276}, where topological heterogeneity and definition methods further complicate performance assessment, this analysis focuses on the spatial average performance. Such an approach benchmarks the macroscopic system capacity and effectively isolates the fundamental gains of the offset pointing strategy.

The total number of bits transmitted per link, $N_b$, is directly proportional to the transmission time, $T_{tx}$, and the system bandwidth, $B$. The transmission time is constrained by the node's energy budget and power consumption:
\begin{align}
    T_{tx}(\Lambda, P_{Tx}) = \frac{E_{node}(\Lambda)}{P_{Tx}} = \frac{E_{total}}{N(\Lambda) \cdot P_{Tx}} = \frac{E_{total}}{\mathcal{A} \cdot \Lambda \cdot P_{Tx}}.
\end{align}
Thus, the objective function to be maximized is:
\begin{align}
    N_b(\Lambda, P_{Tx}) = B \times T_{\text{tx}} = \frac{B \cdot E_{total}}{\mathcal{A} \cdot \Lambda \cdot P_{Tx}}.
    \label{eq:total_bits_objective}
\end{align}

For communication to be reliable, the BER must not exceed a predefined threshold, $\text{BER}_{\text{th}}$. The BER is a function of the SNR, which depends on the received power, $P_{Rx}$. Consequently, for any given node density $\Lambda$ and receiver orientation angle $\delta$, there exists a minimum required transmit power, $P_{\text{Tx,min}}$, that satisfies this quality-of-service (QoS) constraint.

It is important to note that our energy model focuses on the active communication phase. While we acknowledge that practical networks are often event-driven with static power overheads, communication power typically exceeds static consumption by orders of magnitude \cite{11077794}, making it the dominant factor during data transfer. Consequently, static power is treated as a constant background overhead, allowing us to isolate and optimize the energy efficiency of the transmission process itself. This approach assumes the network is in an active communication phase, prioritizing the maximization of transmittable data during the node's operational lifespan.

\begin{figure}[!b]
    \centering
    \includegraphics[width=\columnwidth]{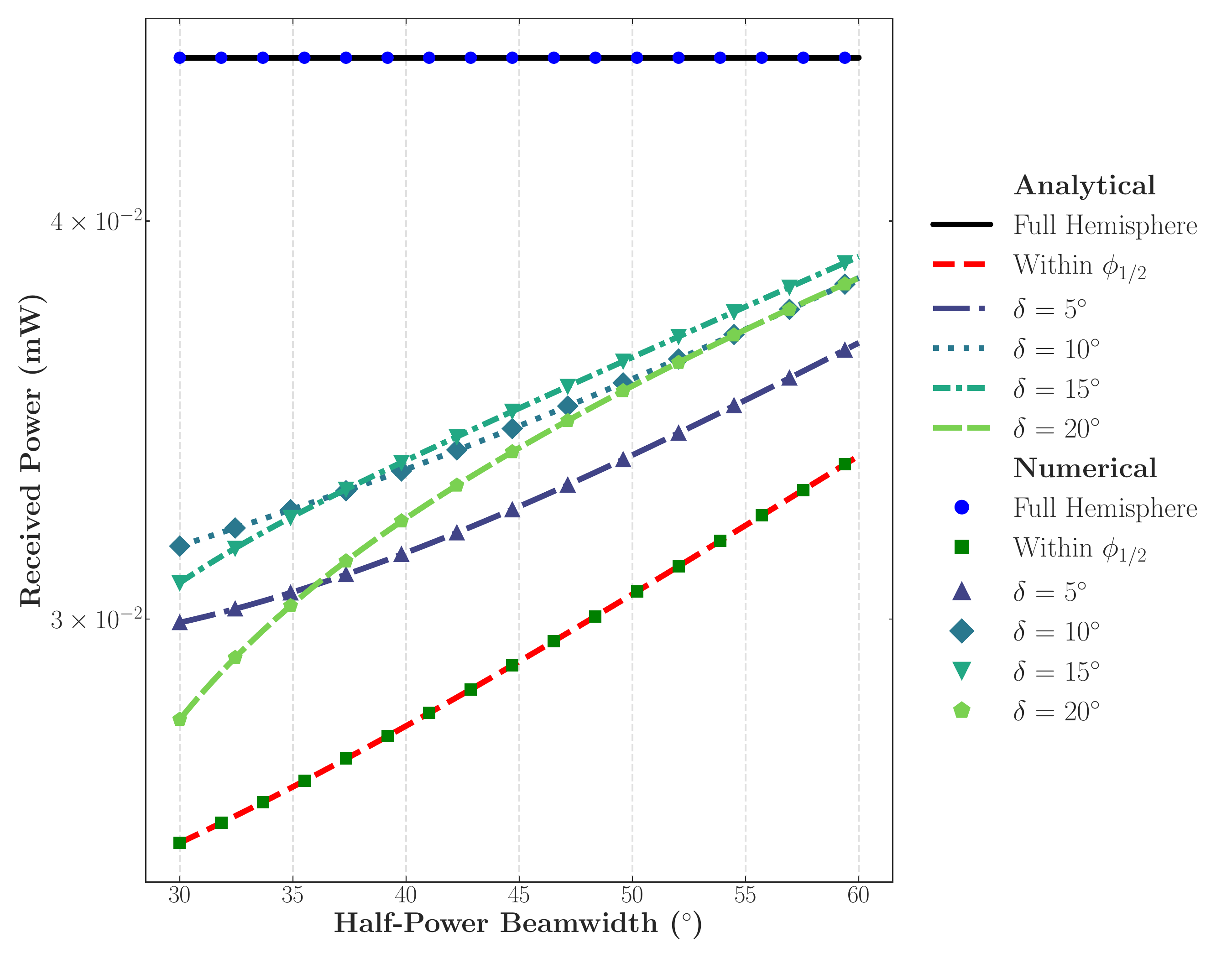}
    \caption{Comparison between closed-form solution and numerical solution for \eqref{eq:40}, \eqref{eq:41}, \eqref{eq:42}}
    \label{fig:Error1}
\end{figure}

To maximize $N_b$ in \eqref{eq:total_bits_objective}, $P_{Tx}$ must be minimized. Therefore, the optimization problem involves operating at precisely this minimum required power, $P_{Tx,min}$. The problem can be formally stated as:
\begin{align}
\begin{aligned}
& \underset{\Lambda, P_{Tx}}{\text{maximize}}
& & N_b = \frac{B \cdot E_{total}}{\mathcal{A} \cdot \Lambda \cdot P_{Tx}} \\
& \text{subject to}
& & \text{BER}(P_{Tx}, \Lambda, \delta) \le \text{BER}_{th}, \\
& & & \Lambda > 0, \\
& & & 0.01 \, \text{W} < P_{Tx} \le P_{Tx,max}.
\end{aligned}
\label{eq:optimization_problem_full}
\end{align}
The lower bound on the transmit power, $P_{Tx} > 0.01 \, \text{W}$, is introduced to account for practical hardware limitations where transmitters have a minimum operational power level.

To quantify the benefits of active receiver alignment, we analyze this optimization problem under two distinct strategies:
\begin{enumerate}
    \item \textbf{Baseline Strategy ($\delta = 0^\circ$):} A non-aligned scenario where the receiver maintains a fixed vertical orientation, making no adjustment for the transmitter's location.
    \item \textbf{Optimized Strategy ($\delta = \delta_{opt}$):} An enhanced scenario where the receiver dynamically adjusts its orientation to an optimal angle, $\delta_{opt}$, that maximizes the received power for the average link geometry.
\end{enumerate}
By solving the optimization problem \eqref{eq:optimization_problem_full} for both strategies, we can evaluate the resultant power savings and the increase in data capacity.

\begin{table}[!t]
    \centering
    \caption{SIMULATION PARAMETERS}
    \begin{tabular}{@{} l c c @{}}
        \toprule
        Parameter & Symbol & Value \\
        \midrule
        SiPM PDE & $\eta$ & 0.31 \\
        SiPM gain & $G$ & \SI{1e6}{A/W} \\
        SiPM FOV & $\phi_{FoV}$ & \SI{120}{\degree} \\
        SiPM cross-talk & $P_{ct}$ & \SI{8}{\percent} \\
        Thermal resistance & $R_L$ & \SI{50}{\ohm} \\
        Bandwidth & $B$ & \SI{1}{Mbps} \\
        Dark current & $I_d$ & \SI{154}{nA} \\
        LED half-angle in optimization problem & $\phi_{1/2}$ & \SI{60}{\degree} \\
        LED wavelength & $\lambda$ & \SI{450}{nm} \\
        LED power & $P_{Tx}$ & \SI{8}{W} \\
        % 1. 在这里使用 \footnotemark，它会自动生成上标数字
        Total beam attenuation\footnotemark & $c(\lambda)$ & \SI{0.151}{m^{-1}}\cite{elfikky2024underwater} \\
        Sea-level solar radiation & $E_{sun}(z=0)$ & \SI{1000}{W/m^2} \\
        Excess noise & $F$ & 1.08 \\
        Directional dependence factor & $L_f$ & 4 \\
        Reflectance of solar radiation & $\zeta_r$ & 1.25 \\
        Bandpass filter window & $\Delta\lambda$ & \SI{50}{nm} \\
        Aperture convex lens & $D$ & \SI{30}{cm} \\  
        Solar radiation attenuation coefficient & $\epsilon$ & 0.2 $\mathrm{m}^{-1}$ \\  
        \bottomrule
    \end{tabular}
    \label{tab:Set up}
\end{table}

\footnotetext{Although the extinction coefficient $c(\lambda)$ varies with environmental conditions, the path loss term $e^{-c(\lambda)L_{dis}}$ acts as a multiplicative factor independent of the receiver orientation $\delta$. Consequently, variations in $c(\lambda)$ scale the magnitude of the received power profile but do not alter the convexity of the objective function with respect to $\delta$. Therefore, the geometric optimality of the proposed offset strategy remains robust to changes in water turbidity. A representative value of $c(\lambda) = 0.151 \, \mathrm{m}^{-1}$ is thus adopted in the simulation without loss of generality.}

\section{Simulation Results}\label{sec:SR}

This section evaluates the performance of our proposed optimization strategy through simulation results. Specifically, we examine how key network parameters influence received power, BER, and overall system efficiency.

The simulation parameters, listed in Table \ref{tab:Set up}, are established based on the specifications of publicly available SiPM detectors and related work \cite{saksvik2025sipm,smart2005underwater,zeng2016survey}.

Moreover, for a comprehensive evaluation, our simulations encompass three distinct levels, which represent a wide spectrum of operational contexts under progressively challenging channel conditions. These levels are defined by the internodal communication distance, which is a primary determinant of path loss:
\begin{itemize}
    \item \textbf{Level 1 (0 -- 50 m):} Representing short-range, high-density applications such as localized sensor grids or communication between tightly-grouped autonomous underwater vehicles (AUVs), this regime is characterized by relatively low path loss.
    
    \item \textbf{Level 2 (50 -- 500 m):} Modeling medium-range links, typical for sparsely deployed underwater infrastructures or communication across large geological features, this level introduces significant channel attenuation as a primary performance-limiting factor.
    
    \item \textbf{Level 3 (500 -- 6000 m):} Encompassing long-range and deep-sea applications, such as vertical links from the seafloor to a surface station, this level presents an extreme challenge and a stress test for system viability due to severe signal attenuation.
\end{itemize}
    Fig.~\ref{fig:Error1} compares the analytical and Monte Carlo simulation results for the received power, $P_{Rx}$, under various analysis modes. The close agreement between the analytical and simulation data validates the accuracy of our proposed model.
    
    %%%%%%%%%%%%%%%%%%%%%%%%%%%%%%%%%%%%%%%%%%%%%%%%%%%%%%%%%%%%%%%%%

    %%%%%%%%%%%%%%%%%%%%%%%%%%%%%%%%%%%%%%%%%%%%%%%%%%%%%%%%%%%%%%%%%

Analyzing the system performance across these three levels (Fig.~\ref{Sim1}, \ref{Sim2}, and~\ref{Sim3}) enables a thorough assessment of the robustness and scalability of our proposed strategies under conditions ranging from benign to highly adverse.
\begin{figure*}[!t] % 使用 figure* 环境，并设置 [t] 或 [b] 控制位置
    \centering
    \includegraphics[width=\textwidth]{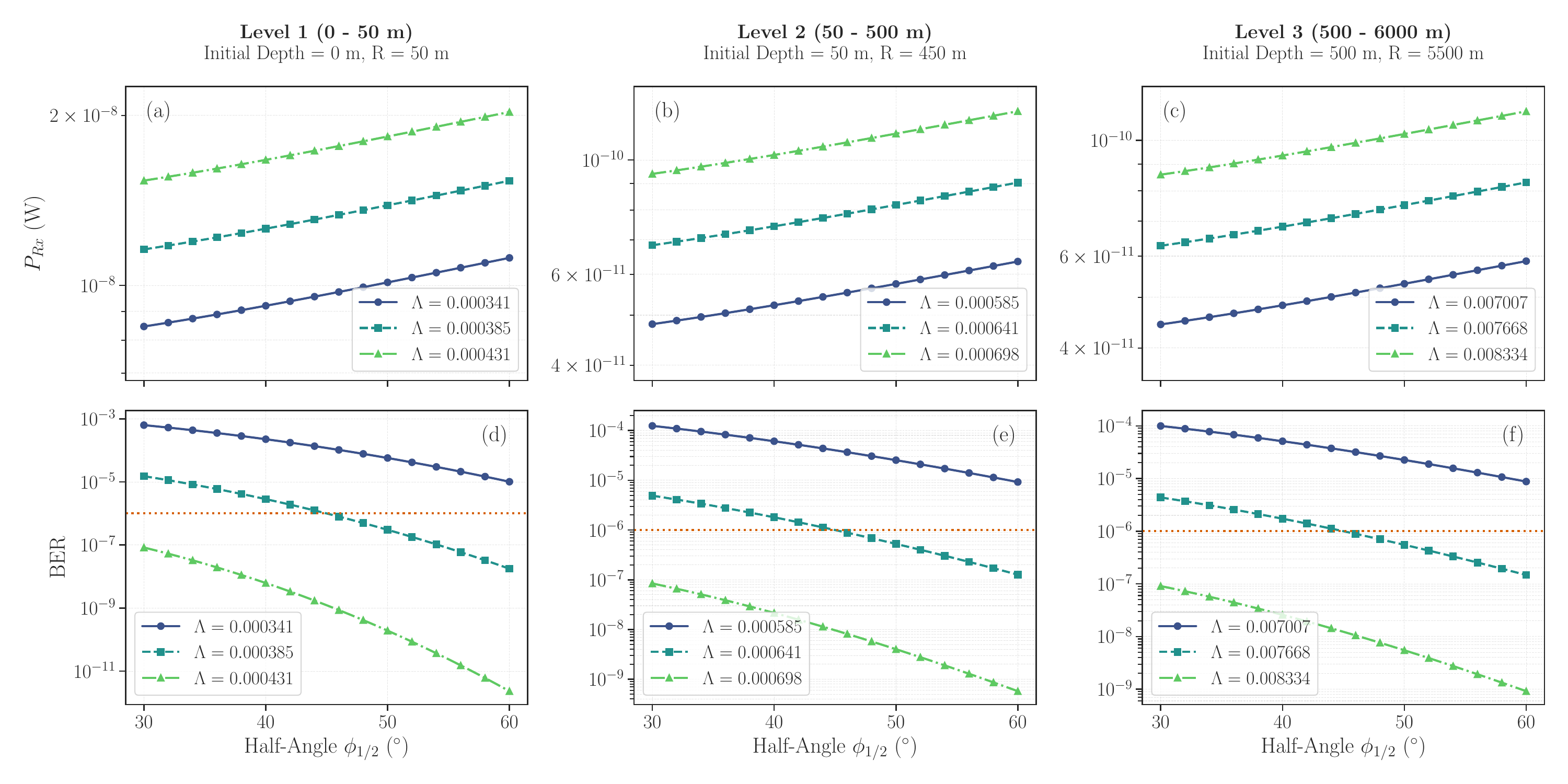} % 设置图片宽度为整页宽度
    \caption{Simulation results showing the effects of $\phi_{1/2}$ on $P_{Rx}$ and BER in different levels ($\delta=0^\circ$)}
    \label{Sim1}
\end{figure*}

Fig.~\subref*{fig:sub1} and~\subref*{fig:sub2} present the normalized received power, $P_{Rx}$, as a function of the half-power beamwidth, $\phi_{1/2}$, and the pointing offset, $\delta$, using heatmaps with corresponding projection plots. These results reveal that introducing a deliberate pointing offset can significantly enhance $P_{Rx}$, and the optimal offset angle is found to be directly proportional to $\phi_{1/2}$. This proportionality is further examined in Fig.~\subref*{fig:sub3}, which displays the relative received power versus the offset angle for four distinct beamwidths, clearly identifying the optimal operational points. These findings provide strong empirical validation for our analytical model.
    
The underlying mechanism for this enhancement lies in the integration of optical power across the receiver's aperture. For a Lambertian source, the irradiance profile is highly non-uniform, peaking at the beam's center and decaying sharply towards the periphery. Consequently, under perfect alignment ($\delta=0^\circ$), the central region of the receiver is maximally irradiated, while the outer regions capture substantially less power. By introducing a small pointing offset, $\delta > 0^\circ$, a beam region with more uniform, albeit lower-peak, irradiance is projected onto the receiver. This strategy sacrifices peak irradiance at the center for a significant gain in power captured at the periphery. As a result, the total received power, the spatial integral of irradiance over the aperture is maximized.

Physically, this method effectively substitutes the low-irradiance beam tail at the receiver's periphery with a higher-irradiance annular section of the beam. An excessive offset, however, will cause the beam to miss the receiver, leading to a sharp decrease in $P_{Rx}$. The positive correlation observed between the optimal offset $\delta$ and the beamwidth $\phi_{1/2}$ arises because broader beams exhibit a more gradual irradiance decay, necessitating a larger geometric offset to align the region of steepest power gradient with the receiver's edge. From a practical standpoint, maintaining perfect transmitter-receiver (Tx-Rx) alignment is often infeasible or prohibitively expensive in dynamic environments subject to factors such as water currents, platform vibrations, and node drift. Our findings advocate for an intentional misalignment strategy where an optimal offset $\delta$ is preset. This approach offers greater robustness to dynamic pointing errors and random jitter compared to the pursuit of perfect alignment, yielding smaller fluctuations in $P_{Rx}$ and thus enhancing link stability and reliability. Furthermore, the relaxed pointing accuracy requirements can facilitate the use of fixed-pointing transceivers, thereby reducing the hardware's size, weight, and power (SWaP) consumption, as well as its cost and deployment complexity.

Fig.~\ref{Sim1} presents the received power, $P_{Rx}$, and BER performance under ideal alignment ($\delta = 0^\circ$) for varying network parameters. The simulations reveal that both metrics improve with increasing node density, $\Lambda$, and transmitter half-power beamwidth, $\phi_{1/2}$. Furthermore, a greater vertical node separation, $R$, necessitates a higher $\Lambda$ to maintain an acceptable BER.

The underlying principle in random underwater networks is that a higher $\Lambda$ reduces the average internodal distance, thereby mitigating path attenuation. This improvement in received signal strength consequently leads to a lower BER. Because the three-dimensional communication path length is determined by both horizontal and vertical node separation ($R$), a large $R$ can substantially elongate the path, even for nodes with close horizontal proximity. Given that signal strength decays significantly with total path length, a substantial $R$ introduces considerable attenuation. To counteract this degradation, a denser network deployment is required to reduce the average horizontal distance, thus keeping the total path length within an effective communication range.
%%%%%%%%%%%%%%%%%%%%%%%%%%%%%%%%%%%%%%%%%%%%%%%%%%%%%%%%%%%%%%%%%
\begin{figure*}[!t] % 使用 figure* 环境，并设置 [t] 或 [b] 控制位置
    \centering
    \includegraphics[width=\textwidth]{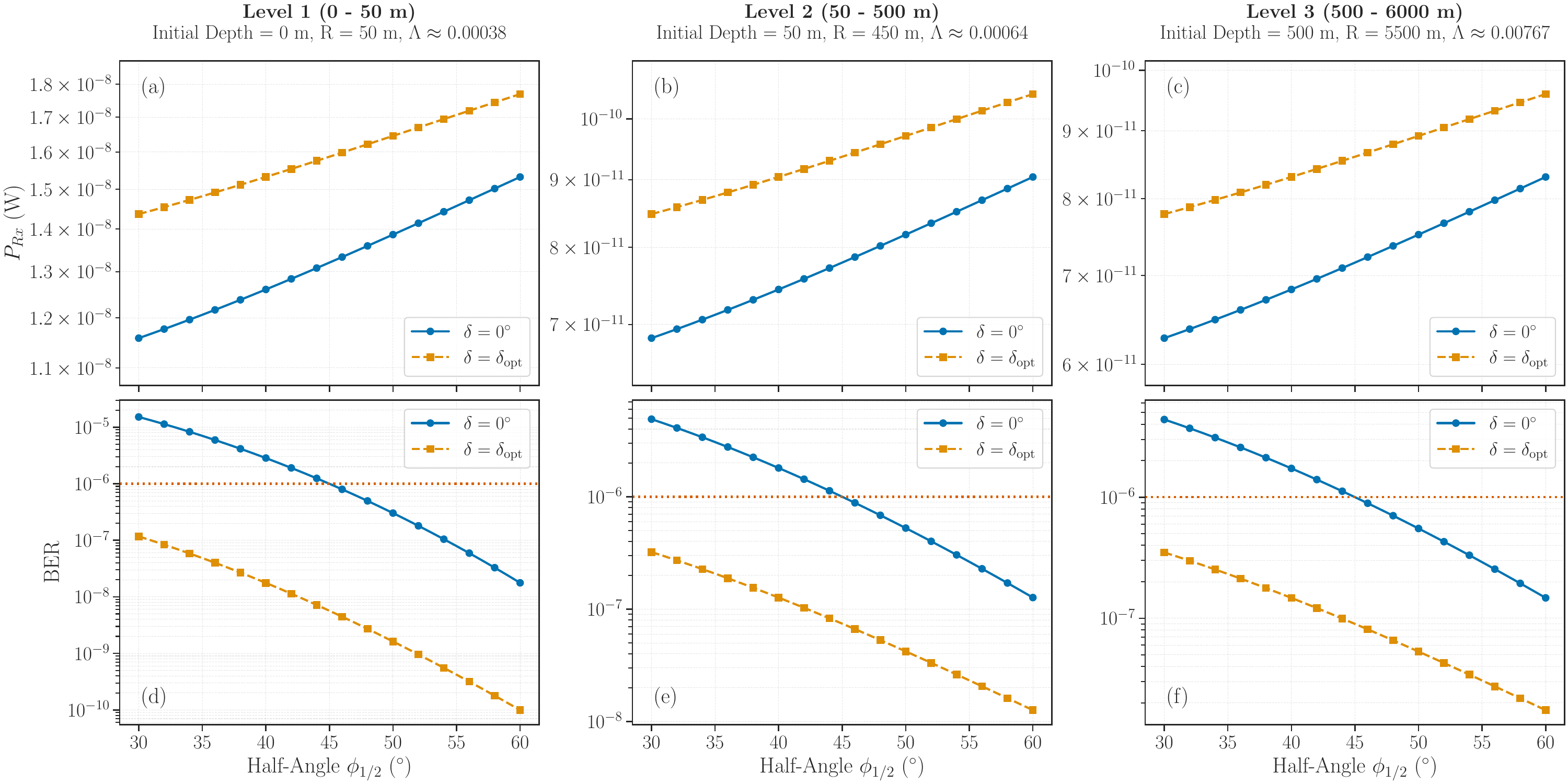} % 设置图片宽度为整页宽度
    \caption{Simulation results showing the effect of $\phi_{1/2}$ on $P_{Rx}$ and BER when $\delta=0^\circ$ and $\delta=\delta_{opt}$ in different levels}
    \label{Sim2}
\end{figure*}
    %%%%%%%%%%%%%%%%%%%%%%%%%%%%%%%%%%%%%%%%%%%%%%%%%%%%%%%%%%%%%%%%%

Fig.~\ref{Sim2} depicts a comparative analysis of $P_{Rx}$ and BER performance for offset-optimized and ideally aligned links under identical node densities ($\Lambda$). The results demonstrate that offset optimization yields an approximately 20\% enhancement in received power, although this gain diminishes with increasing $\phi_{1/2}$, a trend attributed to the more uniform energy distributions of broader beams. The improvement in BER is even more pronounced, typically exhibiting a reduction of one to two orders of magnitude. Moreover, the range of permissible $\phi_{1/2}$ values that satisfy the operational threshold (BER $< 10^{-6}$) is significantly expanded with offset optimization. Collectively, these findings indicate a reduced system sensitivity to variations in $\phi_{1/2}$, resulting in more gradual link degradation despite manufacturing tolerances or operational changes in beamwidth, thereby enhancing hardware robustness.

This enhanced robustness offers significant engineering and economic advantages. For example, improved $P_{Rx}$ and reduced BER enable the use of lower-power transmitters or less sensitive receivers, which are often more cost-effective, while maintaining the same performance level. The relaxed pointing precision simplifies mechanical design and attitude control systems, facilitates deployment, and lowers maintenance costs associated with environmental disturbances. Ultimately, these benefits broaden the operational envelope of UOWC systems, enabling the deployment of cost-effective networks in challenging environments.
\begin{figure*}[!t] % 使用 figure* 环境，并设置 [t] 或 [b] 控制位置
    \centering
    \includegraphics[width=\textwidth]{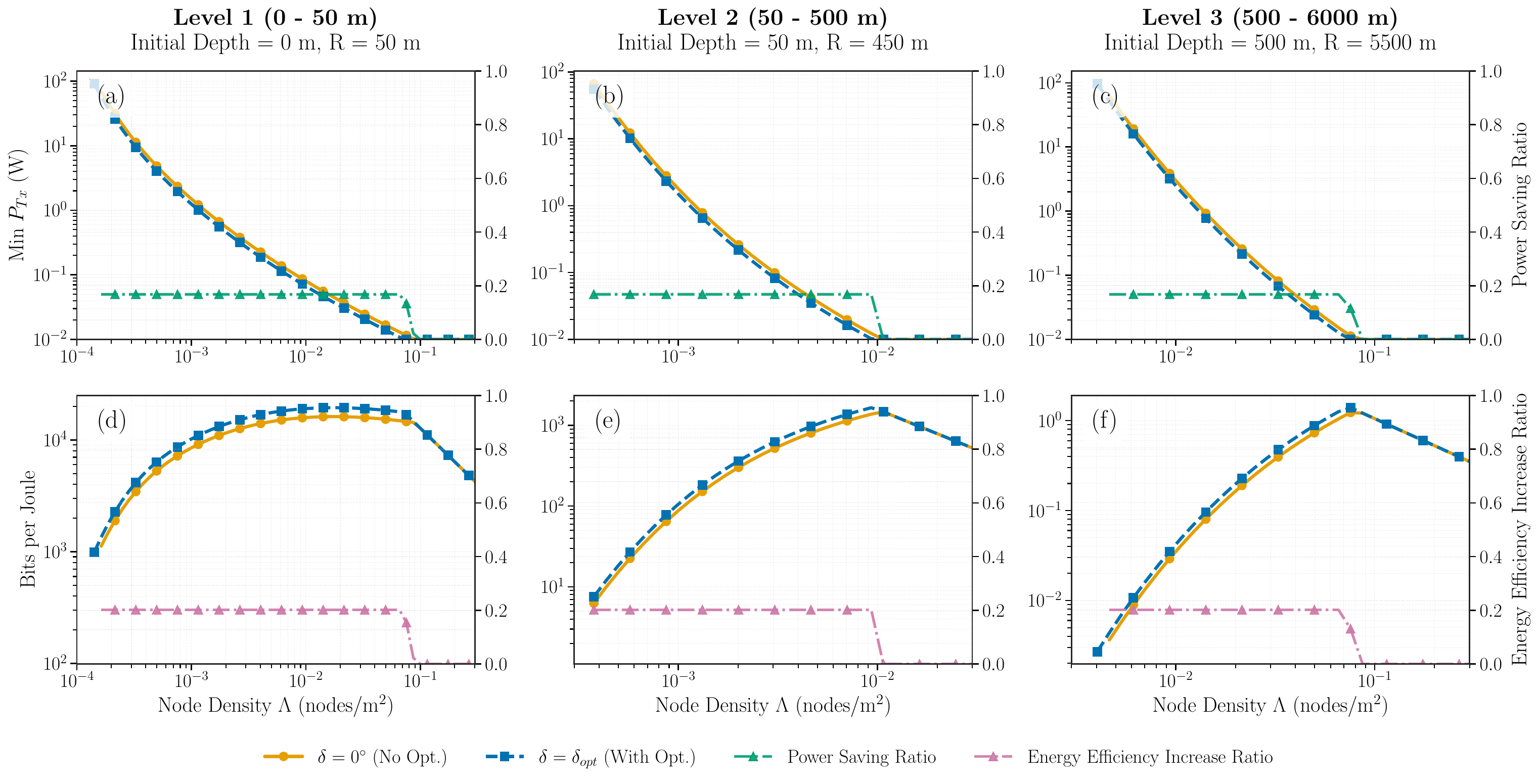} % 设置图片宽度为整页宽度
    \caption{In different levels, when $\delta=0^\circ$ and $\delta=\delta_{opt}$, the relationship between the required $P_{Tx}$ and the corresponding total number of transmitted bits and node density to satisfy BER = $10^{-6}$ ($\phi_{1/2}=60^\circ$)}
    \label{Sim3}
\end{figure*}
    %%%%%%%%%%%%%%%%%%%%%%%%%%%%%%%%%%%%%%%%%%%%%%%%%%%%%%%%%%%%%%%%%
\begin{figure}[!b] % 使用 figure* 环境，并设置 [t] 或 [b] 控制位置
    \centering
    \vspace{-3mm}
    \includegraphics[width=\columnwidth]{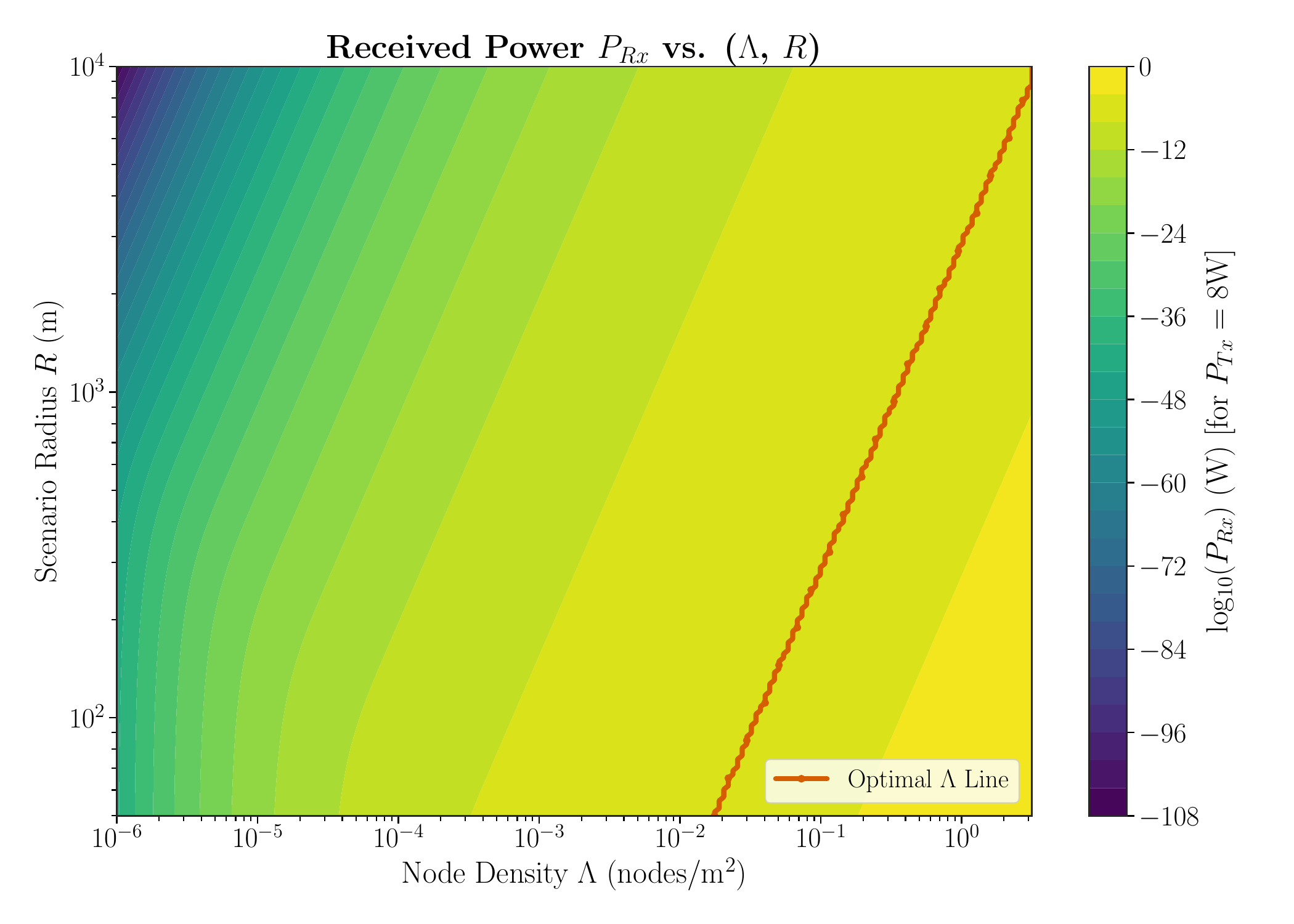} % 设置图片宽度为整页宽度
    \caption{Heatmap of received power $P_{Rx}$ as a function of  $R$ and $\Lambda$ with transition line}
    \label{derive}
\end{figure}
    %%%%%%%%%%%%%%%%%%%%%%%%%%%%%%%%%%%%%%%%%%%%%%%%%%%%%%%%%%%%%%%%%
Fig.~\ref{Sim3} depicts the results of the power optimization problem, governed by the interplay between network physics and a hardware constraint of $P_{Tx} \ge 0.01\,\text{W}$. The results reveal a characteristic dual-phase behavior contingent on node density, $\Lambda$. In Phase 1 (low-density regime), increasing $\Lambda$ substantially reduces the average internodal distance, enabling a sharp decrease in the required transmission power, $P_{Tx}$, and a corresponding increase in total transmittable data. Conversely, Phase 2 (high-density regime) marks a point of diminishing returns, where further network densification yields negligible reductions in path length. In this phase, the escalating energy cost of adding nodes surpasses the marginal power savings per link, causing network throughput and energy efficiency to decline from their theoretical peaks.
    %%%%%%%%%%%%%%%%%%%%%%%%%%%%%%%%%%%%%%%%%%%%%%%%%%%%%%%%%%%%%%%%%

    %%%%%%%%%%%%%%%%%%%%%%%%%%%%%%%%%%%%%%%%%%%%%%%%%%%%%%%%%%%%%%%%%

The key insight lies in the interaction between this phase transition and the hardware power floor. As detailed in the analysis accompanying Fig.~\ref{derive}, the inflection point marking the onset of Phase 2 shifts to a higher $\Lambda$ as the $R$ increases. For levels 2 and 3, characterized by large $R$, this critical density is exceptionally high. However, the optimization process, which reduces $P_{Tx}$ with increasing $\Lambda$, drives the transmission power to its 0.01\,W floor long before this intrinsic inflection point is reached. This hardware-imposed limitation is evident in Fig.~\ref{Sim3}: for levels 2 and 3, the energy efficiency (Bits per Joule) peaks and then declines as soon as $P_{Tx}$ is fixed. This downturn does not signify the onset of Phase 2 but is a direct consequence of the power floor, where per-link power savings cease while the total number of power-consuming nodes continues to rise. In contrast, level 1, with its smaller radius, reaches its natural performance peak at a moderate density before the power constraint becomes active. Therefore, the performance ceiling observed for levels 2 and 3 is a hardware-imposed constraint, not the network's intrinsic limit.

This boundary establishes a critical design guideline for future UOWC evolution. It reveals that increasing network density yields diminishing returns under current hardware capabilities. To unlock the capacity potential of ultra-dense networks, a reduction in the transceiver's baseline power consumption is a prerequisite. Recent advancements in micro-LED technology \cite{MicroLED_Ref} and low-power underwater sensor networks \cite{ali2022recent} suggest that reducing the transmission floor to the milliwatt regime is becoming feasible. Our model quantifies this relationship, serving as a benchmark to determine the necessary hardware specifications required to shift the optimal density threshold and achieve higher aggregate throughput.

For energy-constrained systems like battery-powered underwater sensor networks, this analysis challenges the intuitive notion that more nodes are always better. Instead, it demonstrates the existence of an optimal node density, $\Lambda_{opt}$, that maximizes key performance indicators such as network lifetime or total data throughput. This optimum represents a peak in the return on investment (ROI) for network deployment, informing a more strategic planning process. It allows for a deliberate balance between the initial deployment cost (proportional to the number of nodes) and long-term operational costs (proportional to total energy consumption), thereby optimizing the total lifecycle cost. The offset-pointing strategy consistently enhances this outcome: its nearly 20\% reduction in power consumption can directly translate into a commensurate extension of network lifetime or an increased data capacity under the same energy budget, highlighting its broad applicability and value.

\section{Conclusion and future work}\label{Sec:Conc}
This study presents a significant advancement in the stochastic geometry analysis of UOWC networks. Our contributions include a novel three-dimensional TPPP model that accurately reflects the anisotropic nature of underwater environments, a differential energy analysis framework enabling a more granular understanding of energy consumption dynamics, closed-form expressions for key performance indicators, and an optimization strategy for maximizing data throughput under energy constraints. Simulation results validate the analytical model and demonstrate the substantial benefits of employing optimized receiver orientation strategies, resulting in significant improvements in received power, BER, and overall network capacity. These findings highlight the critical role of considering spatial randomness and energy limitations in the design and deployment of efficient and sustainable UOWC networks. 

In terms of future work, we plan to extend this framework to incorporate layered media models, specifically addressing complex hydrological scenarios such as haloclines and river-to-sea transitions where significant gradients in absorption and scattering exist. Additionally, investigations into different modulation schemes and the development of distributed resource allocation algorithms for multi-node UOWC networks will be pursued. Furthermore, experimental validation of the proposed model and optimization strategies in realistic underwater environments remains a valuable direction.

\bibliographystyle{IEEEtran}
\bibliography{Ref}

\begin{IEEEbiography}
    [{\includegraphics[width=1in,height=1.32in,clip,keepaspectratio]{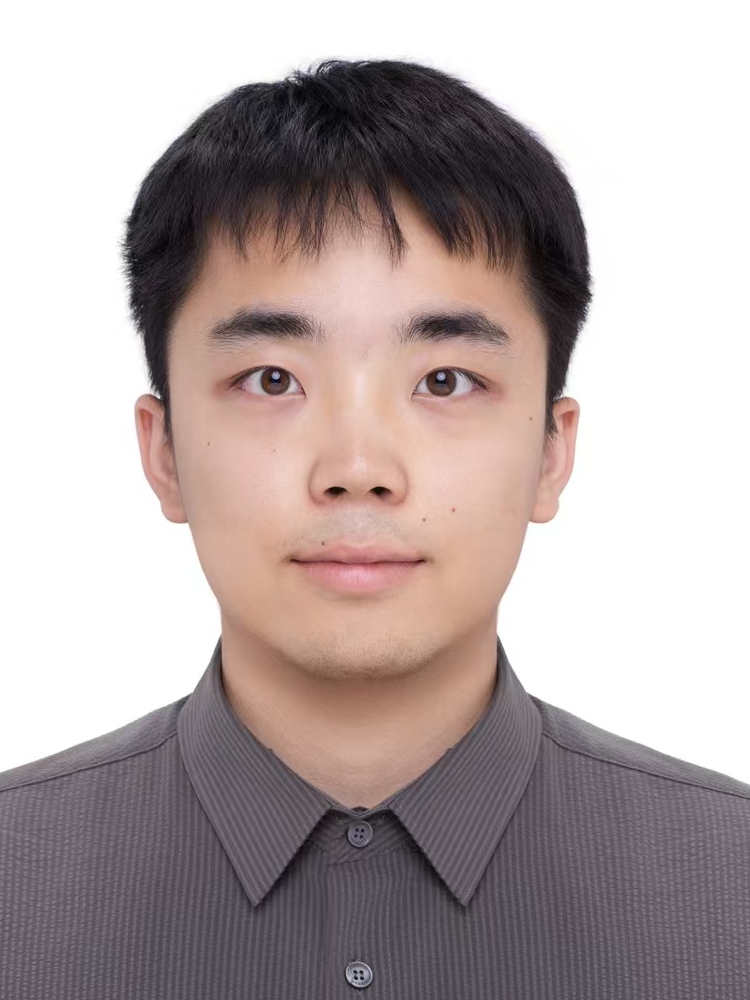}}]{Qiyu Ma}
(Student Member, IEEE) is a undergraduate student in the department of electronic engineering, Tsinghua University, currently a visiting student at the communication theory lab, KAUST. Research focuses on wireless communications, deep learning, and generative AI.
\end{IEEEbiography}
\begin{IEEEbiography}
    [{\includegraphics[width=1in,height=1.32in,clip,keepaspectratio]{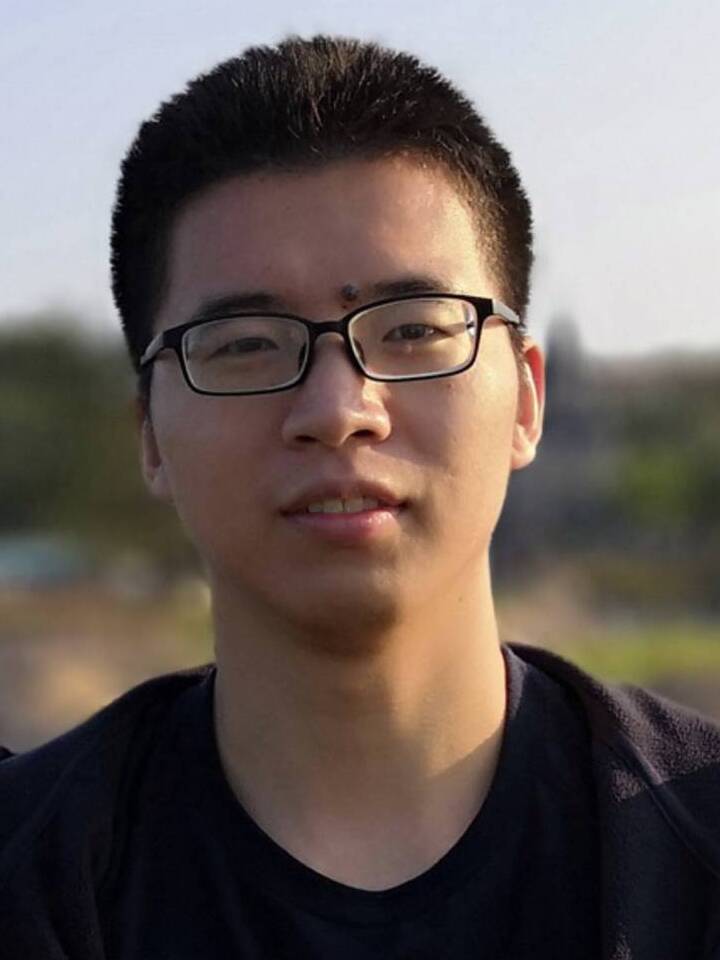}}]{Jiajie Xu}
(Member, IEEE) received his M.Sc. and Ph.D. degrees from Yanshan University and King Abdullah University of Science and Technology (KAUST) in 2019 and 2023. He is a postdoctoral research fellow in the communication theory lab at KAUST. His current research interests include underwater wireless acoustic communication, underwater target detection, underwater cooperative networks, maritime communication, space-air-ground-sea integrated communication systems, joint sensing and communication systems, stochastic geometry, and energy-harvesting wireless networks.
\end{IEEEbiography}
\begin{IEEEbiography}
    [{\includegraphics[width=1in,height=1.33in,clip,keepaspectratio]{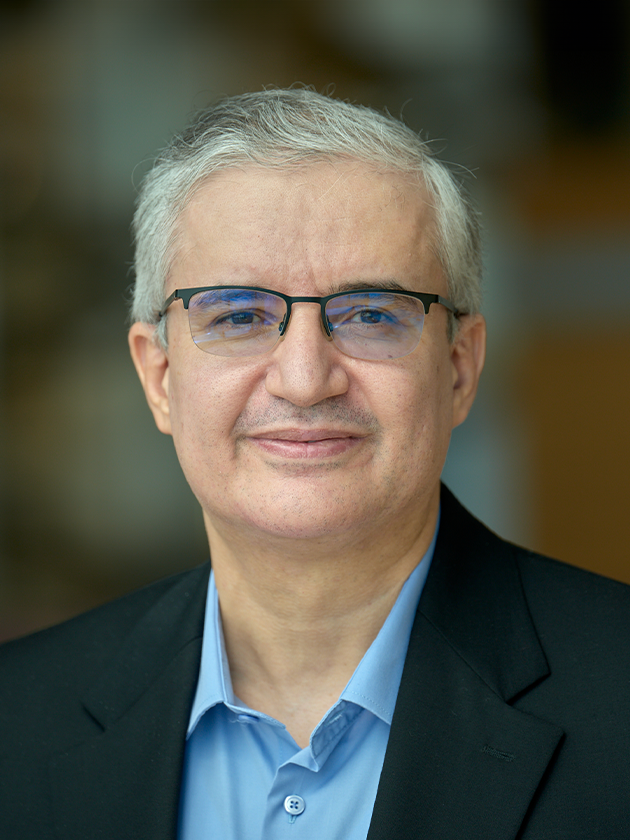}}]{Mohamed-Slim Alouini}
(Fellow, IEEE) received his Ph.D. degree in electrical engineering from the California Institute of Technology, Pasadena, in 1998. He served as a faculty member at the University of Minnesota, Minneapolis, then at Texas A$\&$M University at Qatar, Doha, before joining KAUST as a professor of electrical engineering in 2009. His current research interests include the modeling, design, and performance analysis of wireless communication systems.
\end{IEEEbiography}
\end{document}